\documentclass{article}

\PassOptionsToPackage{numbers, compress}{natbib}

\usepackage[utf8]{inputenc} 
\usepackage[colorlinks = true,
            linkcolor = black,
            urlcolor  = blue,
            citecolor = blue,
            anchorcolor = blue]{hyperref}
\usepackage{url}            
\usepackage{booktabs}       
\usepackage{amsfonts}       
\usepackage{nicefrac}       
\usepackage{microtype}      
\usepackage{xcolor}         
\usepackage{amssymb, amsmath, amsthm, amstext}
\usepackage{verbatim}
\usepackage{bbm}
\usepackage{enumerate}
\usepackage{dsfont}
\usepackage[mathscr]{eucal}
\usepackage{graphicx}
\usepackage{mathrsfs}
\usepackage[ruled,vlined]{algorithm2e}
\usepackage{tikz}
\usepackage{placeins}
\usepackage{subfigure}
\usepackage[final]{neurips_2023}
\newcommand{\Ecomment}[1]{}
\newcommand{\Scomment}[1]{}

\title{Expressive probabilistic sampling in recurrent neural networks}

\author{%
  Shirui Chen\thanks{Corresponding author} \\
  Department of Applied Mathematics\\
  University of Washington, Seattle \\
  \texttt{sc256@uw.edu} \\
  \And
  \hspace{2em}Linxing Preston Jiang \\
  \hspace{2em}Paul G. Allen School of Computer Science \\ 
  \hspace{2em}\& Engineering \\
  \hspace{2em}University of Washington, Seattle \\
  \hspace{2em}\texttt{prestonj@cs.washington.edu} \\
  \And
  \hspace{-3.5em}Rajesh P. N. Rao \\
  \hspace{-3.5em} Paul G. Allen School of Computer Science \\ 
  \hspace{-3.5em}\& Engineering and Center for Neurotechnology \\
  \hspace{-3.5em}University of Washington, Seattle \\
  \hspace{-3.5em}\texttt{rao@cs.washington.edu} 
  \And 
  Eric Shea-Brown \\
  Department of Applied Mathematics\\
  Computational Neuroscience Center \\
  University of Washington, Seattle \\
  \texttt{etsb@uw.edu} \\
}

\begin{document}

\maketitle

\begin{abstract}
  In sampling-based Bayesian models of brain function, neural activities are assumed to be samples from probability distributions that the brain uses for probabilistic computation. However, a comprehensive understanding of how mechanistic models of neural dynamics can sample from arbitrary distributions is still lacking. We use tools from functional analysis and stochastic differential equations to explore the minimum architectural requirements for $\textit{recurrent}$ neural circuits to sample from complex distributions. We first consider the traditional sampling model consisting of a network of neurons whose outputs directly represent the samples ($\textit{sampler-only}$ network). We argue that synaptic current and firing-rate dynamics in the traditional model have limited capacity to sample from a complex probability distribution. We show that the firing rate dynamics of a recurrent neural circuit with a separate set of output units can sample from an arbitrary probability distribution. We call such circuits $\textit{reservoir-sampler networks}$ (RSNs). We propose an efficient training procedure based on denoising score matching that finds recurrent and output weights such that the RSN implements Langevin sampling. We empirically demonstrate our model's ability to sample from several complex data distributions using the proposed neural dynamics and discuss its applicability to developing the next generation of sampling-based Bayesian brain models.
\end{abstract}

\section{Introduction}
There is growing evidence that humans and other animals make decisions by representing uncertainty internally and carrying out probabilistic computations that approximate Bayesian inference \citep{rao_probabilistic_2002,knill_bayesian_2004,doya_bayesian_2007, griffiths_probabilistic_2010, ma_neural_2014}. How networks of neurons in the brain represent probability distributions for Bayesian inference has remained a major open question. There exist two major theories: one assumes that the neural activities encode the parameters of the underlying posterior distributions over sensory stimuli \citep{beck_probabilistic_2008, ma_bayesian_2006, vertes_flexible_2018}. The other is the sampling-based hypothesis, which assumes that the neural responses can be interpreted as samples from a posterior distribution \citep{hoyer_interpreting_2002}. Under this hypothesis, recurrent neural circuits make use of their inherent stochasticity to produce samples from posterior distributions. The sampling-based theory has explained various experimental observations regarding neural variability \citep{echeveste_cortical-like_2020, festa_neuronal_2021, orban_neural_2016}, perceptual decision-making \citep{haefner_perceptual_2016} and spontaneous cortical activity \citep{berkes_spontaneous_2011}. 

Many studies have proposed biologically plausible spiking rules and membrane dynamics models to implement sampling-based probabilistic inference. However, most of these studies mainly consider sampling from discrete Boltzmann distributions \citep{buesing_neural_2011} and multivariate Gaussian distributions \citep{dong_adaptation_2022, aitchison_hamiltonian_2016, masset_natural_2022, hennequin_fast_2014}, only match the first two moments of the distribution \citep{echeveste_cortical-like_2020}, or employ a Monte-Carlo approximation \citep{huang_neurons_2014}. Although these studies use algorithmic substrates that can sample from any distribution with a density function \textit{in theory}, it is not clear whether the underlying neural dynamics are capable of implementing a sufficiently expressive version of these sampling methods. Furthermore, natural image statistics are strongly non-Gaussian \citep{olshausen_emergence_1996}, and experimental evidence shows that humans use non-Gaussian prior representations to support cognitive judgments \citep{houlsby_cognitive_2013, griffiths_optimal_2006}. 
It is known that deep artificial neural networks can be used to generate samples from complex data distributions \citep{song_generative_2019, song_score-based_2021} using a ``U-net'' \citep{ronneberger_u-net_2015} backbone. However, the neural circuits in the cortex are highly coupled with an abundance of recurrent synapses. Therefore, an outstanding question for probabilistic computation in the brain is: what kind of recurrent neural network is capable of efficiently learning to produce samples from an arbitrarily complex posterior distribution? 

In this paper, we study this question under the basic assumption that the dynamics of recurrent neural circuits can be described by stochastic differential equations (SDEs). Note that this assumption is common to a broad range of past research that uses rate-based neural dynamics to implement sampling-based coding\citep{aitchison_hamiltonian_2016, echeveste_cortical-like_2020, dong_adaptation_2022}. Moreover, spike-based models \citep{savin_spatio-temporal_2014, rullan_buxo_sampling-based_2021, masset_natural_2022} that implement balanced spiking networks (BSNs) \citep{boerlin_predictive_2013, buxo_poisson_2020} essentially train spiking networks to follow underlying continuous-time SDEs, so our work applies to this line of research as well (details in Appendix \ref{app:equivalence}).


The contributions of this paper are as follows (Figure \ref{fig:archi}):
\begin{enumerate}



\item We establish the relationship between the sampling power of the neural dynamics and the ability of the dynamics to approximate the score function, which is the gradient of the log probability density function. (Section \ref{sec:score})

\item We show that the synaptic current dynamics of a network of neurons whose outputs directly represent the samples (traditional sampler-only network) is only able to approximate score functions that are in a finite-dimensional function space. (Section \ref{subsec:synap}, Proposition \ref{prop:lim_exp})

\item We prove that the firing rate dynamics of our proposed reservoir-sampler network can sample from a distribution whose score function can approximate that of arbitrary target distributions (with mild restrictions) to arbitrary precision (Section \ref{subsec:fr}, Theorem \ref{thm:main}). 

\item We derive a computationally efficient and biologically-plausible learning rule for our proposed model (Section \ref{subsec:training}) that sidesteps the demands of backpropagation through time, and we empirically demonstrate how our model can sample from several complex data distributions (Section \ref{sec:exp}).



\end{enumerate}

And interpretation of our contributions in biological terms is as follows:  1. Flexible synaptic connectivity within a circuit itself is not enough to allow that circuit to flexibly produce arbitrary patterns of variability involving all of its neurons. 2. In order for the circuit to achieve full flexibility in its output patterns, there need to be hidden variables involved, e.g. states of neurons in an upstream brain area or possibly non-synaptic signaling. If this condition is met, concrete and fairly efficient plasticity rules may be capable of shaping the output patterns as desired. 

\begin{figure}[htb!]
  \centering
  \includegraphics[width=\linewidth]{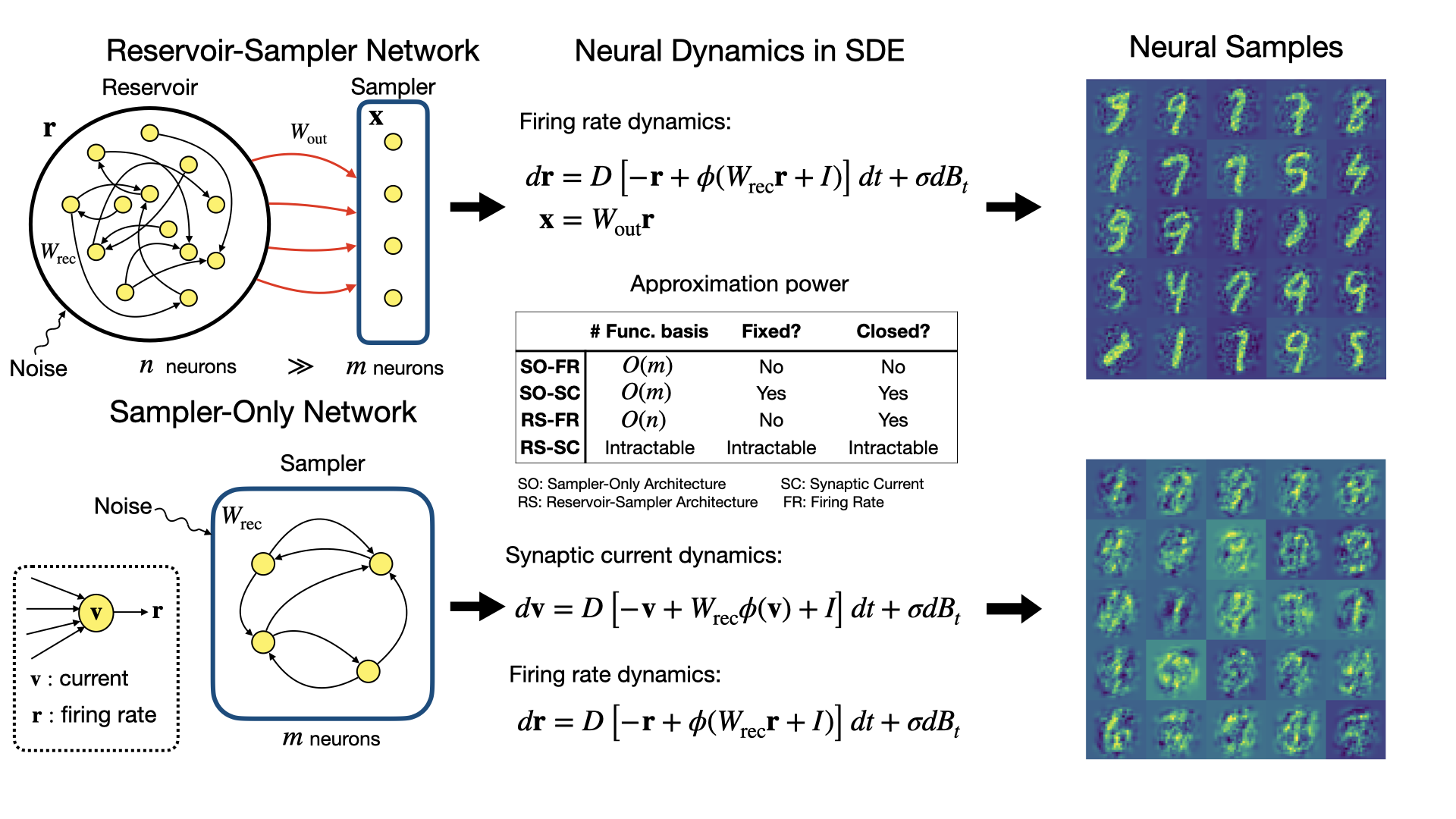}
  \caption{\textbf{Reservoir-Sampler Networks versus Traditional Sampler-Only Networks.} The sampler neurons that produce samples from the target distribution are present for both the sampler-only (SO) network and the reservoir-sampler (RS) network. The reservoir neurons (shown within the large black circle) are only present in the RS network. We explore the firing rate (FR) dynamics (Section \ref{subsec:fr}) for both networks and the synaptic current (SC) dynamics (Section \ref{subsec:synap}) only for the sampler-only network because the stationary distribution of the output neurons is intractable in the RS-SC case. We evaluate the approximation power of the function set represented by the drift term of the neural dynamics, and list the number of basis functions that span the set, whether these basis functions are fixed, and whether the function set is closed under addition (shown in the table).  }
  \label{fig:archi}
\end{figure}

\section{Background} \label{sec:bg}
\subsection{Fokker-Planck equation and stationary distribution}
We consider a general time-homogeneous SDE with drift vector $\m{b}(X_t)$ and diffusion matrix $\Sig = \f{1}{2} \sigma \sigma^T$:
\begin{equation}
    dX_t = \m{b}(X_t) dt + \sig dB_t
\end{equation}
where $\sig\in \R^{n\times m}$ is the diffusion coefficient, and $B_t$ is an $m$-dimensional standard Wiener process. The Fokker-Planck equation of this SDE describes the time evolution of the probability density $p(x, t)$ for a given SDE, assuming the initial density $p(x, 0)$ is known:

\begin{equation} \label{eq:fp_sc}
    \partial_t p = \nabla \cdot (\Sig\nabla p - \m{b}p)
\end{equation}

where $\partial_t = \partial/\partial t$, $\nabla\cdot$ is the divergence operator, and $\nabla$ is the gradient operator. A stationary probability density function is one for which the right-hand-side of equation \eqref{eq:fp_sc} is 0. Mild regularity conditions guaranteeing the existence and uniqueness of a stationary density function are discussed in, for example, \citet{cerrai_second_2001}. We assume these are satisfied by the SDEs considered in this paper. Moreover, to ensure ergodicity and a well-defined score function, we assume that $p$ is supported on $\R^n$, i.e. $p(x)>0$ for all $x\in \R^n$. 

As a special case, consider the following Langevin dynamics, in which the drift term is given by the gradient of the log stationary probability density $\nabla p(x)$, 
\begin{equation} \label{eq:langevin}
  dX_t = \nabla \log p(X_t) dt + \sqrt{2} dB_t.
\end{equation}
It can be verified through the Fokker-Planck equation that $p(x)$ in the dynamics above is indeed a stationary probability density function of the dynamics. Therefore the Langevin dynamics can sample from the distribution $p(x)$ as $t\to \inf$. While the drift term in Langevin dynamics is a gradient vector field, this is not true in general and is often not the case for recurrent neural dynamics (Proposition \ref{prop:1} in Appendix \ref{app:1}). 

\subsection{Score-based generative modeling}
If we would like a particular dynamics to implement Langevin dynamics, we need to fit the drift term of an SDE to the score $s_\tet(\mx) = \nabla \log p(\mx)$ of the probability distribution that we are trying to sample from. This procedure of fitting the score function $s_\tet(\mx)$ is called score matching. In this section, we give a brief summary of one of the major methods of score matching that we will use in this paper, denoising score matching (DSM) \citep{vincent_connection_2011}. The general idea of DSM is to match the score of a noise-perturbed version of the target distribution. More specifically, the explicit score matching loss (left-hand side) and the denoising score matching loss (right-hand side) are related to each other as follows \citep{vincent_connection_2011}:
\begin{equation}\label{eq:DSM}
  \bbE_{q_\sig(\tilde{\mx})}\left[\f{1}{2}\Norm{s_\tet(\tilde{\mx}) - \nabla_{\tilde{\mx}} \log q_\sig (\tilde{\mx})}^2\right] = \bbE_{q_\sig(\tilde{\mx}, \mx)}\left[\f{1}{2}\Norm{s_\tet(\tilde{\mx}) - \nabla_{\tilde{\mx}} \log q_\sig (\tilde{\mx}| \mx)}^2\right] + C_\sig
\end{equation} 
where $\tilde{\mx}$ is a noise-perturbed version of $\mx$, so $q_\sig(\tilde{\mx}|\mx) = \cN(\mx, \sig)$, and $C$ is a constant depending on $\sig$ but \textit{not} on $\tet$. When the noise is 0, i.e. $\tilde{\mx} = \mx$, the left-hand side of the equation above is the explicit score matching loss. Although in theory, we can start from extremely small Gaussian noise, and directly optimize the right hand of equation \eqref{eq:DSM}, empirically it is beneficial to start with large Gaussian noise and gradually decrease the noise magnitude until $q_\sig (\mx) \approx p(\mx)$ and $s_\tet(\mx)\approx \nabla_\mx\log q_\sig(\mx) \approx \nabla_\mx \log p(\mx)$. Specifically, this has been shown to stabilize the training and improve score estimation \citep{song_generative_2019}.

\section{Methods} \label{sec:methods}

\subsection{Do we really need to match the score?} \label{sec:score}
The Langevin dynamics provides an elegant way to construct an SDE given a specific stationary distribution. However, as noted above, the neural network dynamics seldom have a drift term that is a gradient field (Appendix \ref{app:1}). A natural question is therefore whether an SDE with a drift term that is \textit{not} a gradient field (also known as irreversibility) also gives rise to a specific stationary distribution. The answer requires us to look at the Fokker-Planck equation:
\begin{equation}\label{eq:fp}
    \partial_t p(\m{v}, t) = \nabla \cdot (\Sig\nabla p - pF_\tet(\m{v}))
\end{equation}
where $p(\m{v}, t)$ is the probability density function of the variable of interest $\m{v}$ at time $t$, $F_\tet(\m{v})$ is the drift term of the neural dynamics parametrized by $\tet$, and $\Sig = \f{1}{2} \sig\sig^T$ is the diffusion matrix. Since the right-hand side of \eqref{eq:fp} needs to be 0 for a given stationary distribution $p(\m{v})$, we have
\begin{equation}\label{eq:fp_stat}
    \nabla \cdot (\Sig\nabla p - pF_\tet(\m{v})) = 0.
\end{equation}
Therefore $G:= -\Sig\nabla p + pF_\tet(\m{v})$ needs to be a divergence-free (DF) vector field. In other words, $pF_\tet(\m{v})$ is unique up to a DF field given a fixed stationary distribution $p$. \citet{ma2015complete} shows that there exists a skew-symmetric matrix $Q$ such that the DF field can be written as $G = Q\nabla p + p \sum_j \f{\partial}{\partial \m{v}_j} Q_{ij}$, however, this does not shed more light on how expressive $\set{F_\tet}_\tet$ needs to be without more knowledge of $Q$ and its derivative. We show below that under certain conditions, the DF field $G$ can be regarded as a component that is orthogonal to the score function. Therefore the function space $\set{F_\tet}_\tet$ needs to have enough basis functions so that (when projected) it is able to approximate the score function of the target distribution.


We first note that it would be convenient if $F_\tet(\m{v})$ could approximate the gradient fields $\Sig^{-1}\nabla\log p$ for any $p$, in which case $G  =0$. To find out if this is a necessary condition for the dynamics to sample from an arbitrary target distribution $p$, we let $\Sig = \m{I}$ and invoke the Helmholtz-Hodge decomposition (HHD) \citep{bhatia_helmholtz-hodge_2013, majda_vorticity_2002}. The decomposition theorem states that any sufficiently smooth vector field in $L^2(\R^n;\R^n)$ can be uniquely decomposed into a DF vector field and a pure gradient field. In other words, the function space of all DF fields $G$ and the function space of all gradient fields $\nabla p$ are the orthogonal complement of each other. Therefore the projection of $pF_\tet(\m{v})$ onto the subspace of smooth gradient fields still needs to be able to approximate $\nabla p$ despite the freedom to choose arbitrary DF field $G$. 

For example, in the 1-D case, if we assume that both $pF_\tet(\m{v})$ and $\nabla p$ are square-integrable (so they vanish at infinity), then the divergence-free vector field $G$ (which must be constant in 1-D) is 0, therefore $F_\tet(\m{v}) = \nabla  \log p$. As a result, $\set{F_\tet}_\tet$ indeed needs to be able to approximate $\nabla \log p$ for every $p$. In higher dimensions, the same conclusion holds under the assumption of a strict orthogonality constraint: 
\begin{proposition}\label{prop:strict_ortho}
  Let $p$ be the stationary distribution of the neural dynamics, and the diffusion matrix be the identity matrix. If the DF field $G$ is \textit{strictly orthogonal} to the gradient field $\nabla p$, meaning that $G(\m{v})\cdot \nabla p(\m{v}) = 0$ for all $\m{v}$,  then the drift term $F_\tet(\m{v})$ can be written as the sum of a divergence-free field $p^{-1}G$ and a gradient field $\nabla \log p$. 
\end{proposition}


The proof and further detail is in Appendix \ref{app:strict_ortho}. The form of the above decomposition coincides with that of the HHD \citep{bhatia_helmholtz-hodge_2013, bhatia_natural_2014}. Therefore if we enforce the \textit{normal-parallel} boundary condition \citep{chorin_mathematical_1993} for the gradient component, the HHD theorem \citep{bhatia_helmholtz-hodge_2013} says that the orthogonal projection of the function space $\set{F_\tet}_\tet$ onto the space of gradient fields is the function space of gradient fields $\set{\nabla \log p}_p$ satisfying the boundary condition given the strict orthogonality constraint on $G$. 
The upshot is that $\set{F_\tet}_\tet$ needs to admit enough basis functions (Appendix \ref{app:strict_ortho}). Note that the boundary condition will not be restrictive if we take a sufficiently large bounded region. Therefore it is essential for the neural dynamics to have an expressive functional form that is able to approximate complex score functions, even if their dynamics are not gradient fields.

Previous work has rigorously established that the strict orthogonality constraint holds, in particular, when the nonlinear dynamics is linearized around fixed points of $\m{v}$ (where the drift term is zero \citep{ao_potential_2004, kwon_structure_2005, qian_zeroth_2014}).  As a consequence, the conditions of Proposition \ref{prop:strict_ortho} are true locally around fixed points. 

In what follows, we assume that the conditions of the Proposition \ref{prop:strict_ortho} hold. Under this assumption, without loss of generality, we set $G$ to be $\m{0}$ in the following text and explore whether $\Sig^{-1} F_\tet$, which is determined by the specific neural dynamics (Figure \ref{fig:archi}), is able to approximate complex score function $\nabla \log p$.


\subsection{Synaptic current dynamics: sampler-only networks with limited capacity}\label{subsec:synap} 
We consider the following stochastic synaptic current dynamics \citep{dayan_theoretical_2005} (cf. eq. 7.39) that describe a recurrent neural network in terms of the synaptic current that each neuron receives: 
\begin{equation} \label{eq:cureq}
    d \m{v} = D(-\m{v} + W \phi(\m{v}) + I) dt + \sig d\m{B}_t := \msc(\m{v}) dt + \sig d\m{B}_t
\end{equation}
where $\msc(\m{v}):=D(-\m{v} + W \phi(\m{v})+I)$, $\m{v} = [v_1, \cdots, v_m]^T\in \R^m$ is the synaptic current of the $m$ neurons in the recurrent network, $D\in \R^{m\times m}$ is a diagonal matrix where diagonal elements are the decay constants, $d_i = \tau_i^{-1}$, $W\in \R^{m\times m}$ is the connection matrix, $\phi(\cdot)$ is a nonlinear transfer function\footnote{We later need it to be non-polynomial in Theorem \ref{thm:main}},  $I\in \R^m$ is the external input, $\sig\in \R^{m\times l}$ is the diffusion coefficient and $\m{B}_t$ is an $l$-dimensional standard Wiener process. The diffusion term can be interpreted as input from other brain areas (due to the large number of incoming connections, the assumption of Gaussianity can be justified by the central limit theorem \citep{gerstner_neuronal_2014}). We assume that $\tet = \set{D, W, I}$ are tunable parameters through biological learning.
To show the limited expressivity of $\msc$, we have the following corollary from the Hilbert projection theorem showing that $\msc$ is only able to approximate functions in a finite-dimensional function space.

\begin{proposition} \label{prop:lim_exp}
    Let $H:\R^m\to \R^m$, a function in the Hilbert space $L_2(\R^m,\R^m;p)$. Let $\Pi$ be the orthogonal projection operator onto the vector subspace 
    \begin{equation*}
      E = \set{A\m{v} + B\phi(\m{v}) + I | A,B\in \R^{m\times m}, I\in \R^{m\times 1}}
    \end{equation*}
    If $\Norm{H - \Pi H}>0$, then $\Inf_\tet\Norm{H(\m{v}) - \msc(\m{v})} \geq \Norm{(1 - \Pi) H}>0$. 
\end{proposition}
Proof of the proposition is given in Appendix \ref{app:2}. The proposition says that no matter how the parameter of $\msc$ is tuned, the difference between $\msc$ and the target function cannot approach 0, and the lower bound of the error is given by the norm of the component in the target function that is orthogonal to the finite-dimensional function space $E$. Therefore, synaptic current dynamics have a limited ability to match the score function and hence a limited ability to sample from complex probability distributions under the strict orthogonality constraint in Section \ref{sec:score}. The conclusion holds even if we let the diffusion coefficient $\sig$ be tunable. Since $\Sig^{-1}$ is linear,  $\set{\msc}_\tet$ share the same set of basis functions as $\set{\Sig^{-1}\msc}_{\tet, \sig}$. As we will see in the next section, the firing rate dynamics of a recurrent neural circuit with a separate output layer (a reservoir-sampler network) and a learnable diffusion coefficient $\sig$ can sample from arbitrary stationary distributions.

\subsection{Firing-rate dynamics could be a universal sampler} \label{subsec:fr}
\subsubsection{Sampler-only networks}
In this section, we consider the firing rate dynamics \citep{dayan_theoretical_2005} (cf. eq. 7.11) that describe a recurrent neural circuit in terms of the firing rates of the neurons. We first consider the sampler-only network: 
\begin{equation}\label{eq:3}
    d\m{r} = D(-\m{r} + \phi(\wrec\m{r}+I))dt + \sigma dB_t := \mfr(\m{v}) dt + \sig d\m{B}_t.
\end{equation}
Here, $\mfr(\m{v}) = D(-\m{r} + \phi(\wrec\m{r}+I))$ and $D$ is a diagonal matrix with decay constants as its diagonal elements. The stationary solution of the corresponding Fokker-Planck equation satisfies 
\begin{equation} \label{eq:fp_fr}
    \nabla \cdot (\Sig\nabla p - p\mfr) = 0
\end{equation}
where $\Sig := \f{1}{2} \sigma \sigma^T$ is symmetric positive definite (SPD) and $\mfr = D(-\m{r} + \phi(\wrec\m{r}+I))$. Here $\tet = \set{D, W_\mathrm{rec}, I}$ are tunable parameters. If $\Sig$ is invertible, equation \eqref{eq:fp_fr} becomes $\nabla\cdot (\Sig(\nabla p - p \Sig^{-1} \mfr))$. Therefore if $\Sig^{-1}\mfr = \nabla \log p$ is a gradient field, then the stochastic dynamics of the recurrent neural network described by equation \eqref{eq:3} have a stationary distribution $p^*$ such that the score of this distribution $\nabla \log p^* = \Sig^{-1}\mfr$. 

Compared to the synaptic current dynamics where we could only have functional basis $\m{v}_i$ and $\phi(\m{v})_i$, we can now freely choose the functional basis spanning $\set{\mfr}_\tet$ depending on $\wrec$ and $I$, but since there is no linear term before the nonlinear transformation $\phi$, the function set $\set{\mfr}_\tet$ is not closed under addition. If we view $\Sig^{-1} \mfr$ as a neural network with one hidden layer, the number of hidden neurons must be the same as the input dimension, and the diffusion matrix $\Sig$ (hence $\Sig^{-1}$) is restricted to be an SPD matrix. Therefore, we do not get universal approximation power from $\Sig^{-1}\mfr$, and combined with results in Section \ref{subsec:synap}, we see that an RNN by itself does not intrinsically produce samples from arbitrary distributions. As we will see below, if we let a population of output neurons receive inputs from a large reservoir of recurrently connected neurons (a reservoir-sampler network), we are able to obtain samples from complex distributions from the output neurons.

\subsubsection{Reservoir-sampler networks}
Now we consider the reservoir-sampler network where there is a linear readout layer $\wo\in \R^{m\times n}$ of the reservoir whose dynamics is given by equation \eqref{eq:3} (see also the upper row of Figure \ref{fig:archi}). As a special case of Ito's lemma, we have
\begin{equation} \label{eq:4}
  \begin{split}
    \wo d\m{r} = d \wo \m{r} &= \wo \mfr dt + \wo\sigma dB_t\\
    &= (-\wo D\m{r} + \wo D\phi(\wrec \m{r} + I)) dt + \wo\sigma dB_t.
  \end{split}
\end{equation}
Now we assume that $\wrec$ is the product of $\twr$ and $\wo$, i.e. $\wrec = \twr \wo$ and $D = \alp \m{I}$ is a scaled identity matrix. If we denote the output of the recurrent neural network as $\m{x}:= \wo \m{r}\in \R^m$, we derive the following stochastic dynamics for output neurons: 
\begin{equation} \label{eq:out_fp}
  \begin{split}
    d\m{x} &= (-\alp \m{x} + \alp \wo \phi(\wt{W}_{rec} \m{x}+I)) dt + \wo \sigma dB_t := \tmfr(\m{x}) dt + \tilde{\sig} dB_t\\
  \end{split}
\end{equation}
where $\tmfr(\m{x}) = (-\alp \m{x} + \alp \wo \phi(\twr \m{x}+I))$ and $\tilde{\sig} = \wo \sigma $. Therefore in order for the output neurons to sample from a stationary distribution $p$, we need $s_\beta(\m{x})=(\f{1}{2} \tilde{\sig} \tilde{\sig}^T)^{-1}\tmfr(\m{x}):= \wt{\Sig}^{-1} \tmfr(\m{x})$ to match the score $\nabla \log p(\m{x})$. Here $\beta = \set{\twr, \wo, I, \sig}$ are tunable parameters. Additionally, we assume that $\tmfr$ is $\m{0}$ outside a reasonable range of $\m{x}$. This assumption is used to prevent $s_\beta(\m{x})$ from behaving wildly outside the bounded region on which $s_\beta(\m{x})$ has the expressivity to match the score. The following theorem proves that with a large enough number of reservoir neurons, the score-matching loss can be arbitrarily small. The proof is given in Appendix \ref{app:main}.
\begin{theorem} \label{thm:main}
Suppose that we are given a probability distribution with continuously differentiable density function $p(\m{x}): \R^m\to \R^+$ and score function $\nabla \log p(\m{x})$ for which there exist constants $M_1, M_2, a, k>0$ such that 
\begin{align}
    p(\m{x}) &< M_1 e^{-a \Norm{\m{x}}}\\
    \Norm{\nabla \log p(\m{x})}^2 &< M_2 \Norm{\m{x}}^k
\end{align}
when $\Norm{\m{x}}>L$ for large enough $L$. Then for any $\eps>0$, there exists a recurrent neural network whose firing-rate dynamics are given by \eqref{eq:out_fp}, whose recurrent weights, output weights, and the diffusion coefficient are given by $\wrec\in \R^{n\times n}$ of rank $m$, $\wo \in \R^{m\times n}$, and $\sig\in \R^{n\times m}$ respectively, such that, for a large enough $n$, the score of the stationary distribution of the output units $s_\tet(\m{x})$ satisfies $\bbE_{\m{x}\sim p(\m{x})}[\Norm{\nabla \log p(\m{x}) - s_\tet(\m{x})}^2]<\eps$. 
\end{theorem}

This theorem says that for any realistic data distribution with a smooth positive density function, there always exists a reservoir of recurrently-connected neurons whose output units give samples from a distribution whose score function approximates that of the data distribution to arbitrary precision. Given the bound on the score matching error, \citet{block_generative_2022} (cf. Theorem 13)  gives bounds in Wasserstein 2-distance between the stationary distribution of the trained recurrent dynamics and the true data distribution. It is also worth noting that the recurrent weight matrix of the neural circuit in the theorem is of low-rank, so regardless of how many neurons there are in the reservoir, we can always find a low-rank recurrent weight matrix such that the output neurons sample from a correspondingly low-dimensional distribution. 

\subsection{An efficient RNN weight learning algorithm} \label{subsec:training}
The training procedure is derived from the proof of Theorem \ref{thm:main}. The main idea is to first train an auxiliary neural network with one hidden layer, then transfer the weights of the auxiliary neural network to the weights of the recurrent network that we are considering. More specifically, we first optimize an auxiliary feedforward neural network with one hidden layer $\wo \phi(\twr \mx + I)$ using backpropagation with the denoising score matching loss \eqref{eq:DSM} such that 
\begin{equation}
  2\alp(\wo \phi(\twr \mx + I) -\mx)\approx \nabla \log p(\mx).
\end{equation} 
Then we can calculate the diffusion coefficient $\sig = \wo^T (\wo\wo^T)^{-1}$ and the real recurrent weights $\wrec = \twr \wo$ accordingly. The noise magnitude added to the data samples is decreased exponentially over the entire training period. Figure \ref{fig:DP} illustrates how the network can gradually learn the score function during training. We refer the readers to Appendix \ref{app:train_code} for more details.

Note that this method of using an auxiliary neural network is much more computationally efficient than directly matching the score of the stationary distribution of the dynamics \eqref{eq:out_fp}, for which the score function, which involves the matrix inverse $(\tilde{\sig} \tilde{\sig}^T)^{-1}$, needs to be recomputed at each optimization step. Furthermore, since the entire training procedure is equivalent to training a feedforward network with one hidden layer, it sidesteps the often challenging temporal computations associated with the Backpropagation Through Time (BPTT) algorithm used to train deterministic RNNs. 

Although we assumed that the divergence-free field $G = 0$ for the purpose of theoretical analysis, in practice, fast sampling is a major concern for implementing sampling-based inference models of the brain \citep{hennequin_fast_2014,aitchison_hamiltonian_2016, echeveste_cortical-like_2020,  masset_natural_2022}, and reversible\Ecomment{clarify relationship between reversibility and $G=0$} stochastic dynamics, i.e. if $G = 0$, are known for their slow sampling speed. Fortunately, there is a straightforward way to extend our framework and train RSNs to implement irreversible dynamics with a non-zero divergence-free field $G$.  This results in improvements in sampling speed (see Appendix \ref{app:accel} for details).

\section{Experimental results} 
\label{sec:exp}
In this section, we present the results from two tasks. First, we try to let the recurrent neural network learn and generate samples from a 1-D double-mode Gaussian mixture distribution and a 2-D mixture of heavy-tailed Laplace distributions. For the second task, we explore whether a reservoir-sample network with firing rate dynamics is able to sample from the distribution of internal representations computed from PCA filtering of image inputs. All dynamics are simulated with the Euler-Maruyama method. See Appendix \ref{app:train_detail} for hyperparameters used and other training details.

\subsection{Learning mixture distributions} \label{subsec:bimodal}

We consider a 1-D Gaussian mixture distribution whose density function is the average of two Gaussian distributions centered at $\pm 1$, i.e. $p_{\mrm{data}}(x) = \f{1}{2} (\cN(-1, 0.25) + \cN(1, 0.25))$. We artificially generate 10000 data points from this distribution and minimize the denoising score-matching loss:
\begin{equation}
  \cL(\wo, \twr, I) = \bbE_{p_{\mrm{data}}(x)} \bbE_{\tx \sim \cN(x, \sig^2)}\left[\Norm{2\alp(\wo \phi(\twr x + I) -x) + \f{\tx - x}{\sig^2}}^2 \right].
\end{equation}
We use $\phi(\cdot) = ReLU(\cdot)$ and set $\alp = 1/2$. See Figure \ref{fig:DP} for the numerical results. It is worth noting that if we use $\tanh$ as the transfer function, the sampler-only networks are able to learn the score function perfectly, as the score function of the Gaussian mixture distribution we considered is exactly spanned by $f_1(x) = x$ and $f_2(x) = \tanh(2x)$. See Appendix \ref{app:tanh_failure} for an example where sampler-only networks fail to learn the score function even if hyperbolic tangent nonlinearity is used.

\begin{figure}[htb!]
  \centering
  \includegraphics[width=\textwidth]{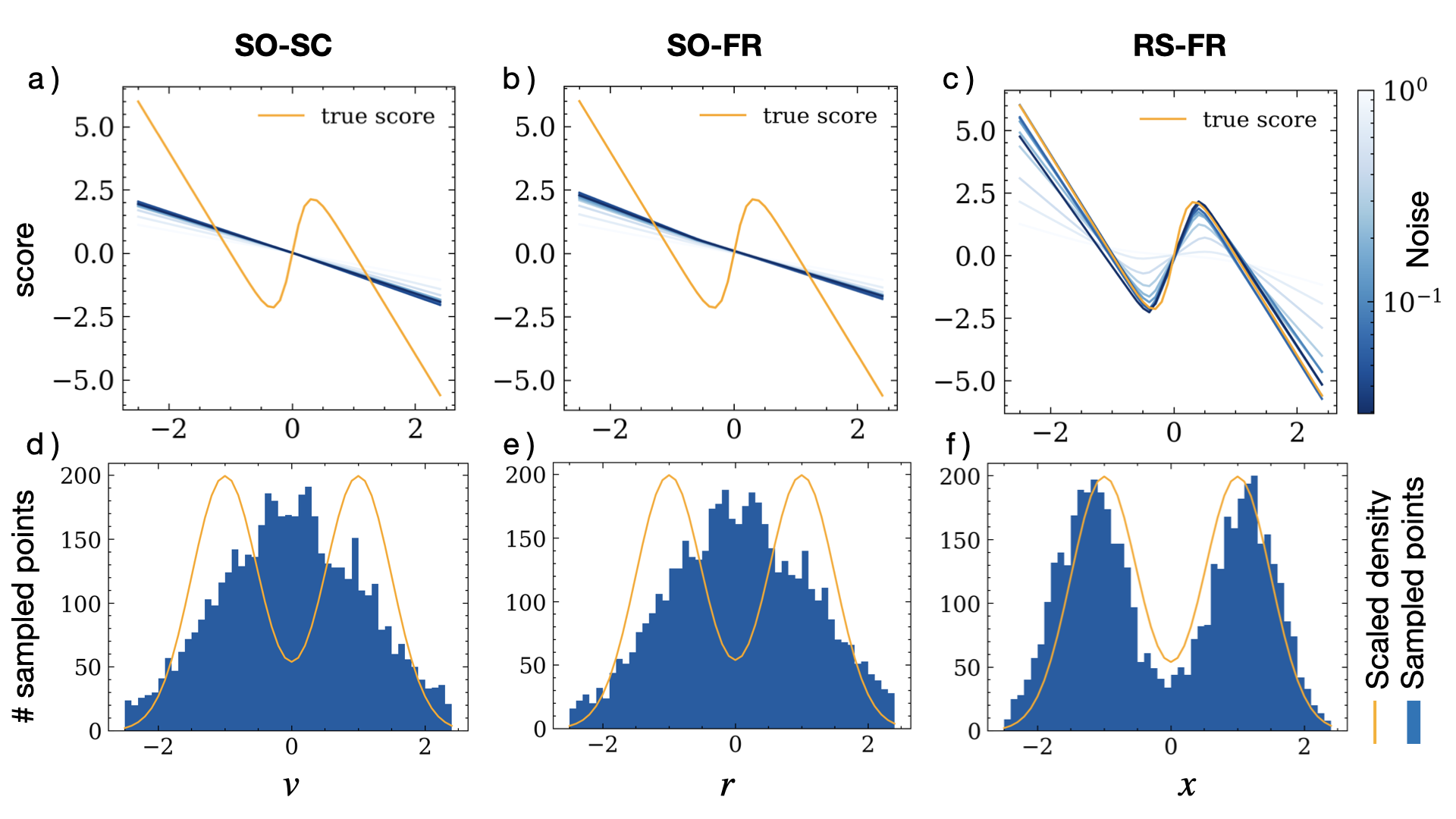}
  \caption{\textbf{Bimodal distribution sampling results.} The 3 tractable cases shown are Sampler-only (SO) networks with both synaptic current (SC) dynamics and firing rate (FR) dynamics and Reservoir-sampler (RS) networks with FR dynamics, which are named SO-SC, SO-FR, and RS-FR respectively. a-c) The score function learned compared to the true score function (orange curve) as we gradually decrease the noise level (the darker the line, the lower the noise level). We see that RS-FR is capable of perfectly fitting the score function, while SO-SC and SO-FR are only able to fit the score function with piecewise linear functions when using the ReLU transfer function. d-f) Histogram of sampled points, and the (scaled) density function of the target distribution. Again the reservoir-sampler network is able to generate samples whose distribution matches the target distribution, while the sampler-only network is not able to do so due to the incorrectly matched score function.}
  \label{fig:DP}
  \end{figure}

Next, we show that the model can learn mixtures of heavy-tailed distributions that are evident in natural image statistics and the neural representations in the primary visual cortex \citep{simoncelli_natural_2001, olshausen_emergence_1996}. We trained the Reservoir-Sampler network with FR dynamics (RS-FR) model on 20000 sampled data points from a 2-D Laplace mixture distribution, whose density is given by
$
    p_{\mrm{data}}(\mx) = \f{1}{2}\left(\mrm{Lap}\left(\m{0}, \begin{bmatrix}
        1 & 0.9 \\
        0.9 & 1 
    \end{bmatrix}\right) + \mrm{Lap}\left(\m{0}, \begin{bmatrix}
        1 & -0.9 \\
        -0.9 & 1 
    \end{bmatrix}\right)
    \right)
$, where $\mrm{Lap}$ denotes the multivariate Laplace distribution. The model successfully learned the probability density of the mixture distribution (Figure \ref{fig:lap} left vs. middle), and captured the heavy tails of the distribution as measured by the kurtosis (Figure \ref{fig:lap} right).

\begin{figure}
    \centering
    \includegraphics[width=\textwidth]{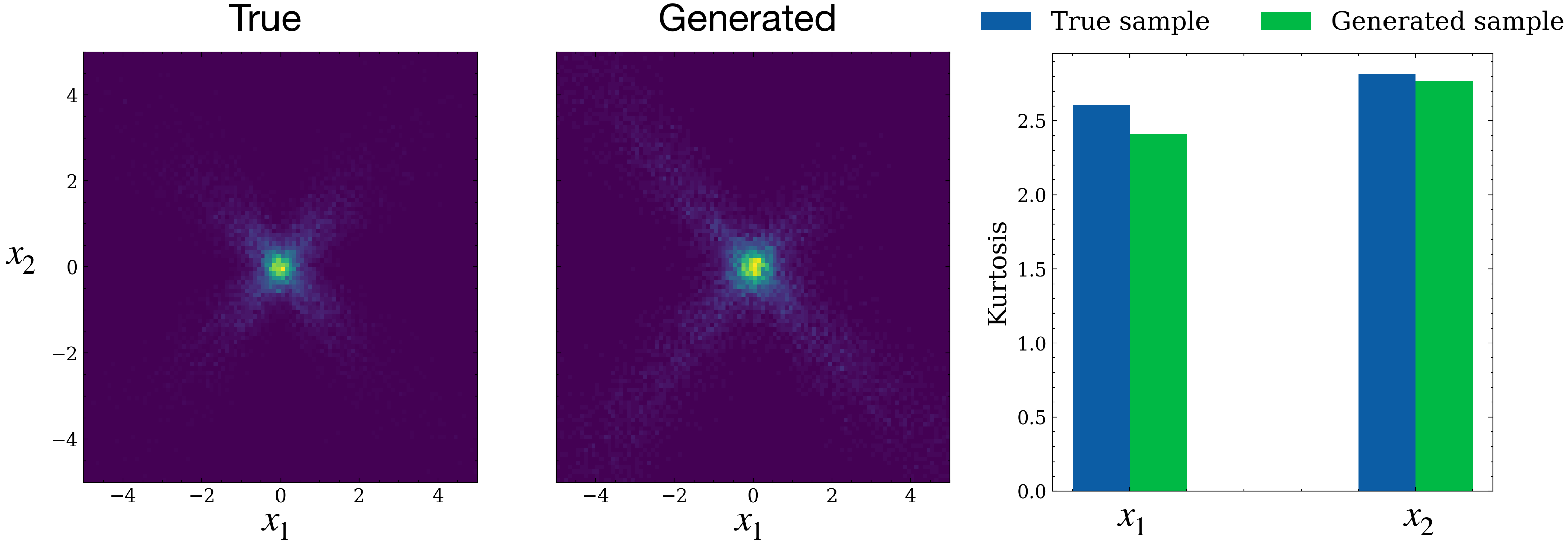}
    \caption{\textbf{RS-FR model learning a mixture of 2-D Laplace distribution.} Left to right: sample density from the true distribution (brighter color denotes higher density); sample density from the learned distribution; marginalized kurtosis of each dimension from the true and learned distribution.}
    \label{fig:lap} 
\end{figure}

\subsection{MNIST generation task} \label{subsec:MNIST}
We also tested the sampling ability of our model using the MNIST dataset \citep{lecun_gradient-based_1998} which contains 60,000 handwritten digits from 0 to 9. We projected MNIST images to a 300-D latent space spanned by the first 300 principal components and trained the weights of the recurrent neural network as described in Section \ref{subsec:training} so that the RNN can sample from the latent distribution. To test the model, we generated images by applying inverse PCA projection to samples generated by the model. The schematics and generated images are shown in Figure \ref{fig:MNIST}. Note that since we are essentially using a shallow network to match the score, we should not expect comparable performance to generative models that use deep ANNs. Our main goal is to illustrate that the reservoir-sampler network using firing rate dynamics is qualitatively more expressive than other traditional neural sampling models (Appendix \ref{app:equivalence}). Finally, we also note that it is highly nontrivial for recurrent neural dynamics to complete such a generative task, and to the best of our knowledge, no previous work has achieved such results. 
\begin{figure}[htb!]
  \centering
  \includegraphics[width=\textwidth]{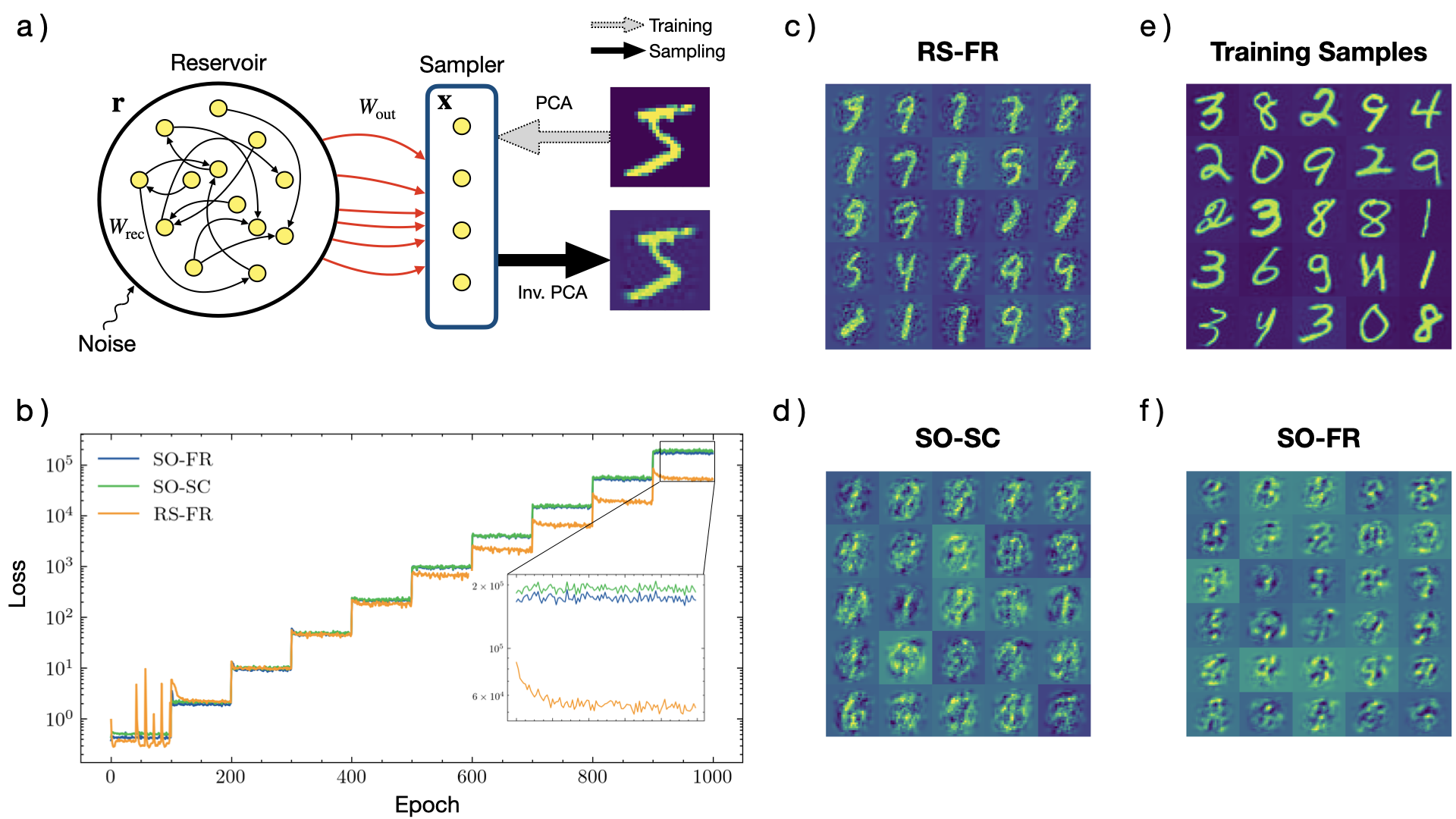}
  \caption{\textbf{Learning to sample the MNIST image distribution.} a) An MNIST image is projected to a 300-D latent space (orange circle) spanned by the first 300 principal components using PCA. The sampler learns to sample from the distribution of this latent space and generates images by applying inverse PCA projection to these samples. The diagram illustrates the RS-FR model. b) The loss curves for 3 different models during training. Every 100 epochs the noise level added to the training samples is reduced (Appendix \ref{app:train_detail}), and the noise increases to a higher value because the score matching loss magnitude depends on the noise level. As shown in the inset, the loss of the RS-FR model decreases throughout the training process when using the lowest fixed noise level. Meanwhile, the losses of the other two models remain unchanged.  c-f) The images generated for the 3 models compared to the digit images generated from latent training samples.}
  \label{fig:MNIST}
\end{figure}

\section{Discussion}\label{sec:diss}
 From the perspective of functional analysis and SDE theory, we prove that under the strict orthogonality constraint, it is essential for neural circuits to have a drift term that has the expressivity to approximate complex score functions, despite the fact that the dynamics do not have to exactly implement the Langevin dynamics. We investigated whether a population of neurons can sample from an arbitrary distribution directly and proved that the drift term of the synaptic current dynamics can only sample from a finite-dimensional function space. Although the drift term of the firing rate dynamics can approximate functions spanned by different basis functions, the number of basis functions is limited. To address this problem, we proposed the reservoir-sampler network for firing rate dynamics. We found that with learnable diffusion coefficients and a sufficiently large reservoir of hidden neurons, the output neurons described using the firing rate dynamics are able to sample from arbitrary data distributions. Our results partly answer the question of what architecture recurrent neural circuits need so that they are able to sample from complex data distributions.  

 Our analysis and empirical experiments affirm the universality of stochastic RNNs. However, this universality comes with limitations. First, we have only analytically shown the existence of weights that enable sampling from complex data distributions; there is no guarantee that one will find such weights through backpropagation. Additionally, in order to obtain the tunable diffusion coefficient during training, a matrix inverse is needed (likewise in the FORCE algorithm \citep{sussillo_generating_2009}). Further, the question of how biological circuits compute the specific gradient and implement the denoising score-matching algorithm remains an open question. Moreover, in our current formulation, we are only able to approximate the score function with a shallow network with one hidden layer. Our preliminary experiments show that one-hidden-layer RSNs cannot readily approximate high-dimensional heavy-tailed distributions (\textit{e.g.}, those of overcomplete sparse coding representations \citep{olshausen_emergence_1996}). It is unclear if this is because of an insufficient number of reservoir neurons. Due to the limitation of the GPU memory, we did not try a higher number of reservoir neurons.


 Our model differs from recent diffusion models \citep{ho_denoising_2020, song_score-based_2021}, which can be seen as time-inhomogeneous SDEs, and has the advantage of being able to run indefinitely in time, making it a suitable candidate for modeling spontaneous activity in the brain. Moreover, while \citet{song_generative_2019} optimizes the denoising score matching loss at different noise levels jointly, we adopt a sequential learning procedure by gradually decreasing the noise level of the training samples. This procedure is more aligned with the developmental processes involved in forming visual representations in the infant brain, where the distribution of visual representations is thought to be noisier (less linearly separable) initially \citep{bayet_temporal_2020}. Our study therefore serves as a starting point for building a mechanistic model for probabilistic computation in the brain that has similar generative power to current AI generative models. 
 
 Biologically, there are multiple ways to interpret the reservoir neurons and sampler neurons in an RSN. First, reservoir and sampler neurons could be seen as different types of neurons in a single brain area, where the dynamics of sampler neurons converge quickly to the equilibrium point. Second, even more straightforwardly, the sampler neurons could be seen as a more separate set of neurons located downstream of the reservoir. We also wish to suggest an alternative interpretation. Biological neural networks are known to have non-synaptic signaling networks (e.g. pervasive neuropeptidergic signaling \citep{bargmann2013connectome}, extensive aminergic signaling \citep{bentley2016multilayer} or potential extrasynaptic signaling \citep{yemini2021neuropal}) in addition to the synaptic connectivity that is typically modeled (i.e., via connection weights).  We suggest that it is possible that the computations of the reservoir may be implemented by non-synaptic networks, and then ``read out'' by certain neurons' spikes.  This possibility is supported by the recent finding of a low correlation between functional activity and the synaptic (``structural'') connectome in C.elegans \citep{yemini2021neuropal}.  Moreover, if we only take the structural connectome into consideration, then the resulting model of C. elegans would correspond to the sampler-only network, which, as our theory predicts, will have limited sampling capability. 

    
    

\section{Conclusion}\label{sec:conclusion}
In this paper, we explore how a recurrent neural circuit can sample from complex probability distributions, an important functional motif in 
probabilistic brain models. We start from a basic assumption that the recurrent neural circuit could be described as an SDE. We show that a recurrently connected neural population by itself has a limited capability to implement stochastic dynamics that can sample from complex data distributions. In contrast, we prove that firing rate dynamics of the output units of a recurrent neural circuit (a reservoir-sampler network) can sample from a richer range of probability distributions. These theoretical results, together with our preliminary experimental results, provide a sufficient condition for neural sampling-based models to exhibit universal sampling capability. Our results therefore provide a foundation for the next generation of sampling-based probabilistic brain models that can explain a wider range of cognitive behaviors.

\section{Acknowledgements}
We are thankful to Profs. Hong Qian, Bamdad Hosseini, and Edgar Walker for their guidance and insight with this project.  We gratefully acknowledge the support of the grant NIH BRAIN R01 1RF1DA055669.

\bibliographystyle{plainnat}
\bibliography{references}

\begin{thebibliography}{53}
\providecommand{\natexlab}[1]{#1}
\providecommand{\url}[1]{\texttt{#1}}
\expandafter\ifx\csname urlstyle\endcsname\relax
  \providecommand{\doi}[1]{doi: #1}\else
  \providecommand{\doi}{doi: \begingroup \urlstyle{rm}\Url}\fi

\bibitem[Aitchison and Lengyel(2016)]{aitchison_hamiltonian_2016}
Laurence Aitchison and Máté Lengyel.
\newblock The {Hamiltonian} {Brain}: {Efficient} {Probabilistic} {Inference}
  with {Excitatory}-{Inhibitory} {Neural} {Circuit} {Dynamics}.
\newblock \emph{PLOS Computational Biology}, 12\penalty0 (12):\penalty0
  e1005186, December 2016.
\newblock ISSN 1553-7358.
\newblock \doi{10.1371/journal.pcbi.1005186}.
\newblock URL
  \url{https://journals.plos.org/ploscompbiol/article?id=10.1371/journal.pcbi.1005186}.
\newblock Publisher: Public Library of Science.

\bibitem[Ao(2004)]{ao_potential_2004}
P.~Ao.
\newblock Potential in stochastic differential equations: novel construction.
\newblock \emph{Journal of Physics A: Mathematical and General}, 37\penalty0
  (3):\penalty0 L25, January 2004.
\newblock ISSN 0305-4470.
\newblock \doi{10.1088/0305-4470/37/3/L01}.
\newblock URL \url{https://dx.doi.org/10.1088/0305-4470/37/3/L01}.

\bibitem[Bargmann and Marder(2013)]{bargmann2013connectome}
Cornelia~I Bargmann and Eve Marder.
\newblock From the connectome to brain function.
\newblock \emph{Nature methods}, 10\penalty0 (6):\penalty0 483--490, 2013.

\bibitem[Bayet et~al.(2020)Bayet, Zinszer, Reilly, Cataldo, Pruitt, Cichy,
  Nelson, and Aslin]{bayet_temporal_2020}
Laurie Bayet, Benjamin~D. Zinszer, Emily Reilly, Julia~K. Cataldo, Zoe Pruitt,
  Radoslaw~M. Cichy, Charles~A. Nelson, and Richard~N. Aslin.
\newblock Temporal dynamics of visual representations in the infant brain,
  August 2020.
\newblock URL
  \url{https://www.biorxiv.org/content/10.1101/2020.02.26.947911v3}.
\newblock Pages: 2020.02.26.947911 Section: New Results.

\bibitem[Beck et~al.(2008)Beck, Ma, Kiani, Hanks, Churchland, Roitman, Shadlen,
  Latham, and Pouget]{beck_probabilistic_2008}
Jeffrey~M. Beck, Wei~Ji Ma, Roozbeh Kiani, Tim Hanks, Anne~K. Churchland, Jamie
  Roitman, Michael~N. Shadlen, Peter~E. Latham, and Alexandre Pouget.
\newblock Probabilistic population codes for {Bayesian} decision making.
\newblock \emph{Neuron}, 60\penalty0 (6):\penalty0 1142--1152, December 2008.
\newblock ISSN 0896-6273.
\newblock \doi{10.1016/j.neuron.2008.09.021}.
\newblock URL \url{https://www.ncbi.nlm.nih.gov/pmc/articles/PMC2742921/}.

\bibitem[Bentley et~al.(2016)Bentley, Branicky, Barnes, Chew, Yemini, Bullmore,
  V{\'e}rtes, and Schafer]{bentley2016multilayer}
Barry Bentley, Robyn Branicky, Christopher~L Barnes, Yee~Lian Chew, Eviatar
  Yemini, Edward~T Bullmore, Petra~E V{\'e}rtes, and William~R Schafer.
\newblock The multilayer connectome of caenorhabditis elegans.
\newblock \emph{PLoS computational biology}, 12\penalty0 (12):\penalty0
  e1005283, 2016.

\bibitem[Berkes et~al.(2011)Berkes, Orban, Lengyel, and
  Fiser]{berkes_spontaneous_2011}
Pietro Berkes, Gergo Orban, Máté Lengyel, and József Fiser.
\newblock Spontaneous {Cortical} {Activity} {Reveals} {Hallmarks} of an
  {Optimal} {Internal} {Model} of the {Environment}.
\newblock \emph{Science}, 331\penalty0 (6013):\penalty0 83--87, January 2011.
\newblock ISSN 0036-8075, 1095-9203.
\newblock \doi{10.1126/science.1195870}.
\newblock URL \url{https://www.science.org/doi/10.1126/science.1195870}.

\bibitem[Bhatia et~al.(2013)Bhatia, Norgard, Pascucci, and
  Bremer]{bhatia_helmholtz-hodge_2013}
Harsh Bhatia, Gregory Norgard, Valerio Pascucci, and Peer-Timo Bremer.
\newblock The {Helmholtz}-{Hodge} {Decomposition}—{A} {Survey}.
\newblock \emph{IEEE Transactions on Visualization and Computer Graphics},
  19\penalty0 (8):\penalty0 1386--1404, August 2013.
\newblock ISSN 1941-0506.
\newblock \doi{10.1109/TVCG.2012.316}.
\newblock Conference Name: IEEE Transactions on Visualization and Computer
  Graphics.

\bibitem[Bhatia et~al.(2014)Bhatia, Pascucci, and Bremer]{bhatia_natural_2014}
Harsh Bhatia, Valerio Pascucci, and Peer-Timo Bremer.
\newblock The {Natural} {Helmholtz}-{Hodge} {Decomposition} for
  {Open}-{Boundary} {Flow} {Analysis}.
\newblock \emph{IEEE Transactions on Visualization and Computer Graphics},
  20\penalty0 (11):\penalty0 1566--1578, November 2014.
\newblock ISSN 1941-0506.
\newblock \doi{10.1109/TVCG.2014.2312012}.
\newblock Conference Name: IEEE Transactions on Visualization and Computer
  Graphics.

\bibitem[Block et~al.(2022)Block, Mroueh, and Rakhlin]{block_generative_2022}
Adam Block, Youssef Mroueh, and Alexander Rakhlin.
\newblock Generative {Modeling} with {Denoising} {Auto}-{Encoders} and
  {Langevin} {Sampling}, October 2022.
\newblock URL \url{http://arxiv.org/abs/2002.00107}.
\newblock arXiv:2002.00107 [cs, math, stat].

\bibitem[Boerlin et~al.(2013)Boerlin, Machens, and
  Denève]{boerlin_predictive_2013}
Martin Boerlin, Christian~K. Machens, and Sophie Denève.
\newblock Predictive {Coding} of {Dynamical} {Variables} in {Balanced}
  {Spiking} {Networks}.
\newblock \emph{PLOS Computational Biology}, 9\penalty0 (11):\penalty0
  e1003258, November 2013.
\newblock ISSN 1553-7358.
\newblock \doi{10.1371/journal.pcbi.1003258}.
\newblock URL
  \url{https://journals.plos.org/ploscompbiol/article?id=10.1371/journal.pcbi.1003258}.
\newblock Publisher: Public Library of Science.

\bibitem[Buesing et~al.(2011)Buesing, Bill, Nessler, and
  Maass]{buesing_neural_2011}
Lars Buesing, Johannes Bill, Bernhard Nessler, and Wolfgang Maass.
\newblock Neural {Dynamics} as {Sampling}: {A} {Model} for {Stochastic}
  {Computation} in {Recurrent} {Networks} of {Spiking} {Neurons}.
\newblock \emph{PLOS Computational Biology}, 7\penalty0 (11):\penalty0
  e1002211, November 2011.
\newblock ISSN 1553-7358.
\newblock \doi{10.1371/journal.pcbi.1002211}.
\newblock URL
  \url{https://journals.plos.org/ploscompbiol/article?id=10.1371/journal.pcbi.1002211}.
\newblock Publisher: Public Library of Science.

\bibitem[Buxó and Pillow(2020)]{buxo_poisson_2020}
Camille E.~Rullán Buxó and Jonathan~W. Pillow.
\newblock Poisson balanced spiking networks.
\newblock \emph{PLOS Computational Biology}, 16\penalty0 (11):\penalty0
  e1008261, November 2020.
\newblock ISSN 1553-7358.
\newblock \doi{10.1371/journal.pcbi.1008261}.
\newblock URL
  \url{https://journals.plos.org/ploscompbiol/article?id=10.1371/journal.pcbi.1008261}.
\newblock Publisher: Public Library of Science.

\bibitem[Cerrai(2001)]{cerrai_second_2001}
S.~Cerrai.
\newblock \emph{Second {Order} {PDE}'s in {Finite} and {Infinite} {Dimension}:
  {A} {Probabilistic} {Approach}}.
\newblock Number no. 1762 in Lecture {Notes} in {Mathematics}. Springer, 2001.
\newblock ISBN 978-3-540-42136-8.
\newblock URL \url{https://books.google.com/books?id=pD-JeUpKrXMC}.

\bibitem[Chorin and Marsden(1993)]{chorin_mathematical_1993}
Alexandre~J. Chorin and Jerrold~E. Marsden.
\newblock \emph{A {Mathematical} {Introduction} to {Fluid} {Mechanics}},
  volume~4 of \emph{Texts in {Applied} {Mathematics}}.
\newblock Springer, New York, NY, 1993.
\newblock ISBN 978-1-4612-6934-2 978-1-4612-0883-9.
\newblock \doi{10.1007/978-1-4612-0883-9}.
\newblock URL \url{http://link.springer.com/10.1007/978-1-4612-0883-9}.

\bibitem[Dayan and Abbott(2005)]{dayan_theoretical_2005}
Peter Dayan and Laurence~F. Abbott.
\newblock \emph{Theoretical {Neuroscience}: {Computational} and {Mathematical}
  {Modeling} of {Neural} {Systems}}.
\newblock MIT Press, August 2005.
\newblock ISBN 978-0-262-54185-5.
\newblock Google-Books-ID: fLT4DwAAQBAJ.

\bibitem[Dong et~al.(2022)Dong, Ji, Chu, Huang, Zhang, and
  Wu]{dong_adaptation_2022}
Xingsi Dong, Zilong Ji, Tianhao Chu, Tiejun Huang, Wenhao Zhang, and Si~Wu.
\newblock Adaptation {Accelerating} {Sampling}-based {Bayesian} {Inference} in
  {Attractor} {Neural} {Networks}.
\newblock October 2022.
\newblock URL \url{https://openreview.net/forum?id=Y0Bm5tL92lg}.

\bibitem[Doya et~al.(2007)Doya, Ishii, Pouget, and Rao]{doya_bayesian_2007}
Kenji Doya, Shin Ishii, Alexandre Pouget, and Rajesh Rao.
\newblock Bayesian brain. {Probabilistic} approaches to neural coding.
\newblock January 2007.

\bibitem[Echeveste et~al.(2020)Echeveste, Aitchison, Hennequin, and
  Lengyel]{echeveste_cortical-like_2020}
Rodrigo Echeveste, Laurence Aitchison, Guillaume Hennequin, and Máté Lengyel.
\newblock Cortical-like dynamics in recurrent circuits optimized for
  sampling-based probabilistic inference {\textbar} {Nature} {Neuroscience}.
\newblock \emph{Nature Neuroscience}, 23\penalty0 (9):\penalty0 1138--1149,
  September 2020.
\newblock ISSN 1546-1726.
\newblock \doi{10.1038/s41593-020-0671-1}.
\newblock URL \url{https://www.nature.com/articles/s41593-020-0671-1}.

\bibitem[Festa et~al.(2021)Festa, Aschner, Davila, Kohn, and
  Coen-Cagli]{festa_neuronal_2021}
Dylan Festa, Amir Aschner, Aida Davila, Adam Kohn, and Ruben Coen-Cagli.
\newblock Neuronal variability reflects probabilistic inference tuned to
  natural image statistics.
\newblock \emph{Nature Communications}, 12\penalty0 (1):\penalty0 3635, June
  2021.
\newblock ISSN 2041-1723.
\newblock \doi{10.1038/s41467-021-23838-x}.
\newblock URL \url{https://www.nature.com/articles/s41467-021-23838-x}.
\newblock Number: 1 Publisher: Nature Publishing Group.

\bibitem[Gerstner et~al.(2014)Gerstner, Kistler, Naud, and
  Paninski]{gerstner_neuronal_2014}
Wulfram Gerstner, Werner~M. Kistler, Richard Naud, and Liam Paninski.
\newblock \emph{Neuronal dynamics: {From} single neurons to networks and models
  of cognition}.
\newblock Cambridge University Press, 2014.

\bibitem[Griffiths and Tenenbaum(2006)]{griffiths_optimal_2006}
Thomas~L. Griffiths and Joshua~B. Tenenbaum.
\newblock Optimal predictions in everyday cognition.
\newblock \emph{Psychological Science}, 17\penalty0 (9):\penalty0 767--773,
  September 2006.
\newblock ISSN 0956-7976.
\newblock \doi{10.1111/j.1467-9280.2006.01780.x}.

\bibitem[Griffiths et~al.(2010)Griffiths, Chater, Kemp, Perfors, and
  Tenenbaum]{griffiths_probabilistic_2010}
Thomas~L. Griffiths, Nick Chater, Charles Kemp, Amy Perfors, and Joshua~B.
  Tenenbaum.
\newblock Probabilistic models of cognition: exploring representations and
  inductive biases.
\newblock \emph{Trends in Cognitive Sciences}, 14\penalty0 (8):\penalty0
  357--364, August 2010.
\newblock ISSN 1364-6613.
\newblock \doi{10.1016/j.tics.2010.05.004}.
\newblock URL
  \url{https://www.sciencedirect.com/science/article/pii/S1364661310001129}.

\bibitem[Haefner et~al.(2016)Haefner, Berkes, and
  Fiser]{haefner_perceptual_2016}
Ralf~M. Haefner, Pietro Berkes, and József Fiser.
\newblock Perceptual {Decision}-{Making} as {Probabilistic} {Inference} by
  {Neural} {Sampling}.
\newblock \emph{Neuron}, 90\penalty0 (3):\penalty0 649--660, May 2016.
\newblock ISSN 0896-6273.
\newblock \doi{10.1016/j.neuron.2016.03.020}.
\newblock URL
  \url{https://www.sciencedirect.com/science/article/pii/S0896627316300113}.

\bibitem[Hennequin et~al.(2014)Hennequin, Aitchison, and
  Lengyel]{hennequin_fast_2014}
Guillaume Hennequin, Laurence Aitchison, and Mate Lengyel.
\newblock Fast {Sampling}-{Based} {Inference} in {Balanced} {Neuronal}
  {Networks}.
\newblock In \emph{Advances in {Neural} {Information} {Processing} {Systems}},
  volume~27. Curran Associates, Inc., 2014.
\newblock URL
  \url{https://proceedings.neurips.cc/paper/2014/hash/a7d8ae4569120b5bec12e7b6e9648b86-Abstract.html}.

\bibitem[Ho et~al.(2020)Ho, Jain, and Abbeel]{ho_denoising_2020}
Jonathan Ho, Ajay Jain, and Pieter Abbeel.
\newblock Denoising {Diffusion} {Probabilistic} {Models}.
\newblock In \emph{Advances in {Neural} {Information} {Processing} {Systems}},
  volume~33, pages 6840--6851. Curran Associates, Inc., 2020.
\newblock URL
  \url{https://proceedings.neurips.cc/paper/2020/hash/4c5bcfec8584af0d967f1ab10179ca4b-Abstract.html}.

\bibitem[Houlsby et~al.(2013)Houlsby, Huszár, Ghassemi, Orbán, Wolpert, and
  Lengyel]{houlsby_cognitive_2013}
Neil~M.T. Houlsby, Ferenc Huszár, Mohammad~M. Ghassemi, Gergo Orbán,
  Daniel~M. Wolpert, and Máté Lengyel.
\newblock Cognitive {Tomography} {Reveals} {Complex}, {Task}-{Independent}
  {Mental} {Representations}.
\newblock \emph{Current Biology}, 23\penalty0 (21):\penalty0 2169--2175,
  November 2013.
\newblock ISSN 0960-9822.
\newblock \doi{10.1016/j.cub.2013.09.012}.
\newblock URL \url{https://www.ncbi.nlm.nih.gov/pmc/articles/PMC3898796/}.

\bibitem[Hoyer and Hyvärinen(2002)]{hoyer_interpreting_2002}
Patrik Hoyer and Aapo Hyvärinen.
\newblock Interpreting {Neural} {Response} {Variability} as {Monte} {Carlo}
  {Sampling} of the {Posterior}.
\newblock In \emph{Advances in {Neural} {Information} {Processing} {Systems}},
  volume~15. MIT Press, 2002.
\newblock URL
  \url{https://proceedings.neurips.cc/paper/2002/hash/a486cd07e4ac3d270571622f4f316ec5-Abstract.html}.

\bibitem[Huang and Rao(2014)]{huang_neurons_2014}
Yanping Huang and Rajesh~PN Rao.
\newblock Neurons as {Monte} {Carlo} {Samplers}: {Bayesian} {Inference} and
  {Learning} in {Spiking} {Networks}.
\newblock In \emph{Advances in {Neural} {Information} {Processing} {Systems}},
  volume~27. Curran Associates, Inc., 2014.
\newblock URL
  \url{https://proceedings.neurips.cc/paper/2014/hash/dc58e3a306451c9d670adcd37004f48f-Abstract.html}.

\bibitem[Knill and Pouget(2004)]{knill_bayesian_2004}
David~C. Knill and Alexandre Pouget.
\newblock The {Bayesian} brain: the role of uncertainty in neural coding and
  computation.
\newblock \emph{Trends in Neurosciences}, 27\penalty0 (12):\penalty0 712--719,
  December 2004.
\newblock ISSN 0166-2236.
\newblock \doi{10.1016/j.tins.2004.10.007}.
\newblock URL
  \url{https://www.sciencedirect.com/science/article/pii/S0166223604003352}.

\bibitem[Kwon et~al.(2005)Kwon, Ao, and Thouless]{kwon_structure_2005}
Chulan Kwon, Ping Ao, and David~J. Thouless.
\newblock Structure of stochastic dynamics near fixed points.
\newblock \emph{Proceedings of the National Academy of Sciences}, 102\penalty0
  (37):\penalty0 13029--13033, September 2005.
\newblock \doi{10.1073/pnas.0506347102}.
\newblock URL \url{https://www.pnas.org/doi/10.1073/pnas.0506347102}.
\newblock Publisher: Proceedings of the National Academy of Sciences.

\bibitem[Lecun et~al.(1998)Lecun, Bottou, Bengio, and
  Haffner]{lecun_gradient-based_1998}
Y.~Lecun, L.~Bottou, Y.~Bengio, and P.~Haffner.
\newblock Gradient-based learning applied to document recognition.
\newblock \emph{Proceedings of the IEEE}, 86\penalty0 (11):\penalty0
  2278--2324, November 1998.
\newblock ISSN 1558-2256.
\newblock \doi{10.1109/5.726791}.
\newblock Conference Name: Proceedings of the IEEE.

\bibitem[Leshno et~al.(1993)Leshno, Lin, Pinkus, and
  Schocken]{leshno_multilayer_1993}
Moshe Leshno, Vladimir~Ya. Lin, Allan Pinkus, and Shimon Schocken.
\newblock Multilayer feedforward networks with a nonpolynomial activation
  function can approximate any function.
\newblock \emph{Neural Networks}, 6\penalty0 (6):\penalty0 861--867, January
  1993.
\newblock ISSN 0893-6080.
\newblock \doi{10.1016/S0893-6080(05)80131-5}.
\newblock URL
  \url{https://www.sciencedirect.com/science/article/pii/S0893608005801315}.

\bibitem[Ma and Jazayeri(2014)]{ma_neural_2014}
Wei~Ji Ma and Mehrdad Jazayeri.
\newblock Neural {Coding} of {Uncertainty} and {Probability}.
\newblock \emph{Annual Review of Neuroscience}, 37\penalty0 (1):\penalty0
  205--220, 2014.
\newblock \doi{10.1146/annurev-neuro-071013-014017}.
\newblock URL \url{https://doi.org/10.1146/annurev-neuro-071013-014017}.
\newblock \_eprint: https://doi.org/10.1146/annurev-neuro-071013-014017.

\bibitem[Ma et~al.(2006)Ma, Beck, Latham, and Pouget]{ma_bayesian_2006}
Wei~Ji Ma, Jeffrey~M. Beck, Peter~E. Latham, and Alexandre Pouget.
\newblock Bayesian inference with probabilistic population codes.
\newblock \emph{Nature Neuroscience}, 9\penalty0 (11):\penalty0 1432--1438,
  November 2006.
\newblock ISSN 1546-1726.
\newblock \doi{10.1038/nn1790}.
\newblock URL \url{https://www.nature.com/articles/nn1790}.
\newblock Number: 11 Publisher: Nature Publishing Group.

\bibitem[Ma et~al.(2015)Ma, Chen, and Fox]{ma2015complete}
Yi-An Ma, Tianqi Chen, and Emily Fox.
\newblock A complete recipe for stochastic gradient mcmc.
\newblock \emph{Advances in neural information processing systems}, 28, 2015.

\bibitem[Majda and Bertozzi(2002)]{majda_vorticity_2002}
Andrew~J. Majda and Andrea~L. Bertozzi.
\newblock \emph{Vorticity and incompressible flow}.
\newblock Camb. {Texts} {Appl}. {Math}. Cambridge University Press, Cambridge,
  2002.
\newblock ISBN 978-0-521-63057-3 978-0-521-63948-4.
\newblock \doi{10.1017/CBO9780511613203}.

\bibitem[Masset et~al.(2022)Masset, Zavatone-Veth, Connor, Murthy, and
  Pehlevan]{masset_natural_2022}
Paul Masset, Jacob~A. Zavatone-Veth, J.~Patrick Connor, Venkatesh~N. Murthy,
  and Cengiz Pehlevan.
\newblock Natural gradient enables fast sampling in spiking neural networks.
\newblock October 2022.
\newblock URL \url{https://openreview.net/forum?id=Yopob26XjmL}.

\bibitem[Olshausen and Field(1996)]{olshausen_emergence_1996}
Bruno~A. Olshausen and David~J. Field.
\newblock Emergence of simple-cell receptive field properties by learning a
  sparse code for natural images.
\newblock \emph{Nature}, 381\penalty0 (6583):\penalty0 607--609, June 1996.
\newblock ISSN 1476-4687.
\newblock \doi{10.1038/381607a0}.
\newblock URL \url{https://www.nature.com/articles/381607a0}.
\newblock Number: 6583 Publisher: Nature Publishing Group.

\bibitem[Orban et~al.(2016)Orban, Berkes, Fiser, and
  Lengyel]{orban_neural_2016}
Gergo Orban, Pietro Berkes, József Fiser, and Máté Lengyel.
\newblock Neural {Variability} and {Sampling}-{Based} {Probabilistic}
  {Representations} in the {Visual} {Cortex}.
\newblock \emph{Neuron}, 92\penalty0 (2):\penalty0 530--543, October 2016.
\newblock ISSN 0896-6273.
\newblock \doi{10.1016/j.neuron.2016.09.038}.
\newblock URL \url{https://www.cell.com/neuron/abstract/S0896-6273(16)30639-0}.
\newblock Publisher: Elsevier.

\bibitem[Qian(2014)]{qian_zeroth_2014}
Hong Qian.
\newblock The zeroth law of thermodynamics and volume-preserving conservative
  system in equilibrium with stochastic damping.
\newblock \emph{Physics Letters A}, 378\penalty0 (7):\penalty0 609--616,
  January 2014.
\newblock ISSN 0375-9601.
\newblock \doi{10.1016/j.physleta.2013.12.028}.
\newblock URL
  \url{https://www.sciencedirect.com/science/article/pii/S0375960113011766}.

\bibitem[Rao et~al.(2002)Rao, Olshausen, and Lewicki]{rao_probabilistic_2002}
Rajesh P.~N. Rao, Bruno~A. Olshausen, and Michael~S. Lewicki, editors.
\newblock \emph{Probabilistic models of the brain: {Perception} and neural
  function}.
\newblock Probabilistic models of the brain: {Perception} and neural function.
  The MIT Press, Cambridge, MA, US, 2002.
\newblock ISBN 978-0-262-18224-9.
\newblock \doi{10.7551/mitpress/5583.001.0001}.
\newblock Pages: x, 324.

\bibitem[Ronneberger et~al.(2015)Ronneberger, Fischer, and
  Brox]{ronneberger_u-net_2015}
Olaf Ronneberger, Philipp Fischer, and Thomas Brox.
\newblock U-{Net}: {Convolutional} {Networks} for {Biomedical} {Image}
  {Segmentation}, May 2015.
\newblock URL \url{http://arxiv.org/abs/1505.04597}.
\newblock arXiv:1505.04597 [cs].

\bibitem[Rudin(1991)]{rudin_functional_1991}
W.~Rudin.
\newblock \emph{Functional {Analysis}}.
\newblock International series in pure and applied mathematics. McGraw-Hill,
  1991.
\newblock ISBN 978-0-07-054236-5.
\newblock URL
  \url{https://www.google.com/books/edition/Functional_Analysis/Sh_vAAAAMAAJ?hl=en}.

\bibitem[Rullán~Buxó and Savin(2021)]{rullan_buxo_sampling-based_2021}
Camille Rullán~Buxó and Cristina Savin.
\newblock A sampling-based circuit for optimal decision making.
\newblock In \emph{Advances in {Neural} {Information} {Processing} {Systems}},
  volume~34, pages 14163--14175. Curran Associates, Inc., 2021.
\newblock URL
  \url{https://proceedings.neurips.cc/paper/2021/hash/76444b3132fda0e2aca778051d776f1c-Abstract.html}.

\bibitem[Savin and Denève(2014)]{savin_spatio-temporal_2014}
Cristina Savin and Sophie Denève.
\newblock Spatio-temporal {Representations} of {Uncertainty} in {Spiking}
  {Neural} {Networks}.
\newblock In \emph{Advances in {Neural} {Information} {Processing} {Systems}},
  volume~27. Curran Associates, Inc., 2014.
\newblock URL
  \url{https://proceedings.neurips.cc/paper/2014/hash/4e2545f819e67f0615003dd7e04a6087-Abstract.html}.

\bibitem[Simoncelli and Olshausen(2001)]{simoncelli_natural_2001}
Eero~P Simoncelli and Bruno~A Olshausen.
\newblock Natural {Image} {Statistics} and {Neural} {Representation}.
\newblock \emph{Annual Review of Neuroscience}, 24\penalty0 (1):\penalty0
  1193--1216, 2001.
\newblock \doi{10.1146/annurev.neuro.24.1.1193}.
\newblock URL \url{https://doi.org/10.1146/annurev.neuro.24.1.1193}.
\newblock \_eprint: https://doi.org/10.1146/annurev.neuro.24.1.1193.

\bibitem[Song and Ermon(2019)]{song_generative_2019}
Yang Song and Stefano Ermon.
\newblock Generative {Modeling} by {Estimating} {Gradients} of the {Data}
  {Distribution}.
\newblock In \emph{Advances in {Neural} {Information} {Processing} {Systems}},
  volume~32. Curran Associates, Inc., 2019.
\newblock URL
  \url{https://proceedings.neurips.cc/paper/2019/hash/3001ef257407d5a371a96dcd947c7d93-Abstract.html}.

\bibitem[Song et~al.(2021)Song, Sohl-Dickstein, Kingma, Kumar, Ermon, and
  Poole]{song_score-based_2021}
Yang Song, Jascha Sohl-Dickstein, Diederik~P. Kingma, Abhishek Kumar, Stefano
  Ermon, and Ben Poole.
\newblock Score-{Based} {Generative} {Modeling} through {Stochastic}
  {Differential} {Equations}, February 2021.
\newblock URL \url{http://arxiv.org/abs/2011.13456}.
\newblock arXiv:2011.13456 [cs, stat].

\bibitem[Sussillo and Abbott(2009)]{sussillo_generating_2009}
David Sussillo and L.~F. Abbott.
\newblock Generating {Coherent} {Patterns} of {Activity} from {Chaotic}
  {Neural} {Networks}.
\newblock \emph{Neuron}, 63\penalty0 (4):\penalty0 544--557, August 2009.
\newblock ISSN 0896-6273.
\newblock \doi{10.1016/j.neuron.2009.07.018}.
\newblock URL
  \url{https://www.sciencedirect.com/science/article/pii/S0896627309005479}.

\bibitem[Vincent(2011)]{vincent_connection_2011}
Pascal Vincent.
\newblock A {Connection} {Between} {Score} {Matching} and {Denoising}
  {Autoencoders}.
\newblock \emph{Neural Computation}, 23\penalty0 (7):\penalty0 1661--1674, July
  2011.
\newblock ISSN 0899-7667.
\newblock \doi{10.1162/NECO_a_00142}.
\newblock Conference Name: Neural Computation.

\bibitem[Vértes and Sahani(2018)]{vertes_flexible_2018}
Eszter Vértes and Maneesh Sahani.
\newblock Flexible and accurate inference and learning for deep generative
  models.
\newblock In \emph{Advances in {Neural} {Information} {Processing} {Systems}},
  volume~31. Curran Associates, Inc., 2018.
\newblock URL
  \url{https://proceedings.neurips.cc/paper/2018/hash/955cb567b6e38f4c6b3f28cc857fc38c-Abstract.html}.

\bibitem[Yemini et~al.(2021)Yemini, Lin, Nejatbakhsh, Varol, Sun, Mena, Samuel,
  Paninski, Venkatachalam, and Hobert]{yemini2021neuropal}
Eviatar Yemini, Albert Lin, Amin Nejatbakhsh, Erdem Varol, Ruoxi Sun, Gonzalo~E
  Mena, Aravinthan~DT Samuel, Liam Paninski, Vivek Venkatachalam, and Oliver
  Hobert.
\newblock Neuropal: a multicolor atlas for whole-brain neuronal identification
  in c. elegans.
\newblock \emph{Cell}, 184\penalty0 (1):\penalty0 272--288, 2021.

\end{thebibliography}


\newpage
\appendix
\section{Helmholtz-Hodge decomposition and proof of Proposition \ref{prop:strict_ortho}} \label{app:strict_ortho}
First we give a proof for Proposition \ref{prop:strict_ortho}:
\begin{prop*}
  Let $p$ be the stationary distribution of the neural dynamics, and the diffusion matrix be the identity matrix. If the DF field $G$ is \textit{strictly orthogonal} to the gradient field $\nabla p$, meaning that $G(\m{v})\cdot \nabla p(\m{v}) = 0$ for all $\m{v}$,  then the drift term $F_\tet(\m{v})$ can be written as the sum of a divergence-free field $p^{-1}G$ and a gradient field $\nabla \log p$. 
\end{prop*}
\begin{proof}
    Since $p$ is the stationary distribution of the neural dynamics, from the Fokker-Planck equation we have
    \begin{equation}
        F_\tet(\m{v}) = p^{-1}G + \nabla \log p
    \end{equation}
    where $G$ is divergence free. Now 
    \begin{align*}
        \nabla \cdot (p^{-1} G) &= \nabla p^{-1}\cdot G + p^{-1}  \nabla\cdot G\\
        &= \f{\nabla p \cdot G}{p^2} & \text{G is divergence free}\\
        &= 0. & \text{G and $\nabla p$ is pointwise orthogonal }
    \end{align*}
    Therefore, $p^{-1}G$ is also divergence free.
\end{proof}

Now we give a statement of the Helmholtz-Hodge decomposition, which is an application of the Hodge decomposition theorem to vector calculus on $\R^n$ \citep{bhatia_helmholtz-hodge_2013, majda_vorticity_2002}.
\begin{theorem} [\textbf{Helmoholtz-Hodge Decomposition}] \label{thm:HHD}
    A smooth vector field $\m{V}$ on a bounded domain of $\R^n$ can be uniquely decomposed into the sum of a gradient field $\nabla u$ that is normal to the boundary, a divergence-free field $\m{v}$ that is parallel to the boundary and a harmonic field $\m{h}$, i.e. $\m{V} = \nabla u + \m{v} + \m{h} $. Moreover, the three components are orthogonal to each other in the $L_2$ sense. 
\end{theorem}
 We note that the form of the decomposition in Proposition \ref{prop:strict_ortho} coincides with the Helmoholtz-Hodge decomposition with a zero harmonic field. Therefore, we can define the orthogonal projection operator $\bbP$ that projects $F_\tet$ onto its gradient component $\nabla \log p$. Under the strict orthogonality constraint that requires G to be orthogonal to $\nabla p$ \textit{pointwise}, to be able to sample from a broad range of distributions with $\nabla \log p$ orthogonal to the boundary, we need that for all $p$, there exists $G_p$ and $\tet$ such that $F_\tet (\m{v})= p^{-1}G_p + \nabla \log p$. Applying the projection operator on both sides, we have
 \begin{equation}
     \bbP F_\tet (\m{v})= \bbP (p^{-1}G_p + \nabla \log p) = \nabla \log p.
 \end{equation}
 Moreover, if the function space $\set{F_\tet}_\tet$ admits a set of basis functions $\set{e_i}_i$, i.e. $F_\tet  = \sum_i c_i e_i $, then 
 \begin{equation}
     \nabla \log p = \bbP F_\tet = \sum_i c_i (\bbP e_i).
 \end{equation}
 This means that the space of $\nabla \log p$ needs to be spanned by $\set{\bbP e_i}_i$, so $\set{F_\tet}_\tet$ should also be spanned by an infinite number of basis functions $\set{e_i}_i$ if we want to sample from a complex class of $\nabla \log p$. 
\section{Neural dynamics not following gradient fields}\label{app:1}
For the meaning of notations, see Section \ref{subsec:synap}.
\begin{proposition} \label{prop:1}
When $m\geq 2$, $\msc(\m{v}):=D(-\m{v} + W \phi(\m{v})+I)$ is a gradient field if and only if $W$ is a zero matrix. 
\end{proposition}
\begin{proof}

    ($\implies$) If $\msc(\m{v}) = \nabla G(\m{v})$ is a gradient field, then 
    \begin{equation}
        \f{\pa}{\pa v_j}\msc(\m{v})_i = \f{\pa}{\pa v_i} \msc(\m{v})_j = \f{\pa}{\pa v_i \pa v_j} G(\m{v})
    \end{equation}
    because $\msc(\m{v})$ is smooth, and we can switch the order of partial differentiation. Written in the form of \eqref{eq:cureq}, we have for all $v_i$ and $v_j$ (where $i\neq j$),
    \begin{align*}
        \f{\pa}{\pa v_j} \left(-d_i v_i + d_i\sum_{k = 1}^n W_{ik} \phi(v_k) + d_iI_i\right) &= \f{\pa}{\pa v_i}\left( -d_j v_j + d_j\sum_{k = 1}^n W_{jk} \phi(v_k) + d_jI_j\right)\\
        \f{\pa}{\pa v_j} d_i W_{ij}\phi(v_j) &= \f{\pa}{\pa v_i} d_j W_{ji} \phi(v_i)\\
        d_i W_{ij}  \phi_{v}(v_j) &= d_j W_{ji} \phi_{v}(v_i),
    \end{align*}
    since $v_i$ and $v_j$ are arbitrary and $\phi$ is nonlinear (so the derivative $\phi_v$ is not the same for all $v$), this implies $d_i W_{ij} = d_j W_{ji} = 0$. But the decay constants cannot be 0. Therefore, $W_{ij} = W_{ji} = 0$

    ($\impliedby$) If $W$ is a zero matrix, then $\msc(\m{v}) = \nabla D\left(-\f{\m{v} \cdot \m{v}}{2} + I\cdot \m{v}\right)$ is a gradient field. 
\end{proof}

\subsection{Hilbert projection theorem and proof of Proposition \ref{prop:lim_exp}} \label{app:2}
Here we directly state the Hilbert Projection Theorem. For detailed proof, we refer readers to standard functional analysis textbooks \citep{rudin_functional_1991}.
\begin{theorem}\label{thm:hilbert}(Hilbert Projection Theorem)
    Given a Hilbert space $H$, and a nonempty closed convex set $E\subset H$, for any $x\in H$, there exists a unique $\Pi x \in E$, such that $\Norm{\Pi x - x} = \Inf_{c\in E} \Norm{x - c}$. 
\end{theorem}
Here $\Pi$ is called the orthogonal projection operator. Next, we give the proof of Proposition \ref{prop:lim_exp}.
\begin{proof}
    Here we consider a special case of a Hilbert space $L_2(\R^m,\R^m;p)$,where the inner product is defined by $\inner{f}{g} = \int_{\R^m} \norm{f\cdot g}p(\m{x}) d\m{x}$ and we let $E = \set{A\m{v} + B\phi(\m{v}) + I | A,B\in \R^{m\times m}; I\in\R^{n\times 1}}$. The set of functions $E$ is a closed convex vector subspace of $L_2(\R^m,\R^m;p)$ because $E$ is closed under addition and finite-dimensional. Hence by Hilbert projection theorem, there exists a unique $\Pi H\in E$ such that $\Norm{(1 - \Pi) H} = \Inf _{F\in E} \Norm{H - F}>0$. Because $\set{\msc}_\tet \subset E$ if the decay matrix $D$ is diagonal, we have 
    \begin{equation}
        \Inf _{\msc} \Norm{H - \msc} \geq\Inf _{F\in E} \Norm{H - F} = \Norm{(1-\Pi)H} >0.
    \end{equation}
\end{proof}

\section{Proof of Theorem \ref{thm:main}} \label{app:main}
First, we give a statement of the universal approximation theorem as in \citep{leshno_multilayer_1993}.
\begin{thm}\label{thm:uat} (Universal Approximation Theorem)
    The continuous transfer function $\phi$ is not polynomial if and only if for every $k\in \N, m\in \N$, compact set $U\subset \R^k$, $f \in C(U, \R^m)$, and $\eps >0$, there exists $n\in \N, W_2\in \R^{n\times k}, b\in \R^{n}, W_1\in \R^{m\times n}$ such that 
    \begin{equation}
        \sup_{x\in U} \Norm{f(x) - g(x)} < \eps
    \end{equation}
    where $g(x) = W_1\phi(W_2 x + b)$. Or equivalently, $\Sig_n = \text{span}\set{\phi(\m{w}\cdot \mx + b): \m{w}\in \R^k, b\in \R}$ is dense in $C(\R^k)$ if and only if $\phi$ is not polynomial. 
\end{thm}

The theorem says that for any transfer function that is not polynomial ($\tanh$ for example), there is a multi-layer neural network with the transfer function that is able to approximate any continuous function to arbitrary precision on a closed and bounded subset of $\R^m$. The proof of Theorem \ref{thm:main} proceeds in three steps: 
\begin{enumerate}
    \item Choose a large enough compact support $U$ for $s_\tet(\m{x})$ such that 
    \begin{equation} \label{eq:mthm_1}
    \int_{\R^m\setminus U} p(\m{x}) \Norm{\nabla \log p(\m{x}) - s_\tet(\m{x})}^2 d\m{x} = \int_{\R^m\setminus U} p(\m{x}) \Norm{\nabla \log p(\m{x})}^2 d\m{x} <\f{\eps}{2}
    \end{equation}
    \item Choose a large enough $n$ and appropriate $\tet$ such that the universal approximation theorem can be applied to show that 
    \begin{equation} \label{eq:mthm_2}
        \int_{U} p(\m{x}) \Norm{\nabla \log p(\m{x}) - s_\tet(\m{x})}^2 d\m{x}< \f{\eps}{2}
    \end{equation}
    \item  Combine equation \eqref{eq:mthm_1} and \eqref{eq:mthm_2} so that we have
    \begin{equation}
    \begin{split}
      \bbE_{\m{x}\sim p(\m{x})} [\Norm{\nabla \log p(\m{x}) - s_\tet(\m{x})}^2] = \int_{\R^m} p(\m{x}) \Norm{\nabla \log p(\m{x}) - s_\tet(\m{x})}^2 d\m{x}<\eps
    \end{split}
    \end{equation}
\end{enumerate}
We begin with two lemmas that will be helpful in the proof of the theorem. Both of them are standard results \footnote{see \hyperlink{here}{ https://dlmf.nist.gov/8.11}}, but we include here for completeness. Let the upper incomplete Gamma function be defined by $\Gam(s, C) = \int_C^\inf r^{s-1} e^{-r} dr$.
\begin{lemma} \label{lem:upper_gam}
  If $s$ is a positive integer, then $\Gam(s, C) = (s-1)!e^{-C}\sum_{i = 0}^{s-1} \f{C^i}{i!}$
\end{lemma}

\begin{proof}
  This can be proved by induction. For $s = 1$,
  \begin{equation}
    \Gam(1, C) = \int_C^\inf e^{-r} dr = e^{-C}
  \end{equation}
  Therefore the conclusion holds. For $s>1$, we assume that for $s = k$, the conclusion holds. Then for $s = k+1$, Through integration by parts, we have
  \begin{equation}
    \Gam(k+1, C) = s\Gam(k, C) + C^{k} e^{-C}
  \end{equation}
  Therefore
  \begin{equation}
    \begin{split}
      \Gam(k+1, C) &= k(k-1)!e^{-C}\sum_{i = 0}^{k-1} \f{C^i}{i!} + C^{k} e^{-C}\\
      &= k ! e^{-C} \left(\f{C^k}{k!}+ \sum_{i = 0}^{k-1} \f{C^i}{i!}\right)\\
      &= k! e^{-C} \sum_{i = 0}^{k} \f{C^i}{i!}
    \end{split}
  \end{equation}
  Therefore the conclusion holds for $s = k+1$ as well. By induction, it holds for all positive integer $s$. 
\end{proof}

\begin{lemma} \label{lem:gam_lim}
  $\Gam(s, C)\to 0$ as $C\to \inf$. 
\end{lemma}
\begin{proof}
  When $C\to \inf$, $\Gam (s,C)<\int_C^\inf r^{\ceil{s-1}} e^{-r} dr$. Therefore W.L.O.G., we can assume that $s-1$ is a positive integer. By lemma \ref{lem:upper_gam}, 
  \begin{equation}
    \lim_{C\to \inf}\Gam(s, C) = (s-1)! \sum_{i = 0}^{s-1} \lim_{C\to\inf} \f{e^{-C} C^i}{i!} = 0\\
  \end{equation}
  since the exponential term $e^{-C}$ is beyond all orders. 
\end{proof}
Now we give a formal proof of Theorem \ref{thm:main}.

\begin{thm*} 
  Suppose that we are given a probability distribution with continuously differentiable density function $p(\m{x}): \R^m\to \R^+$ and score function $\nabla \log p(\m{x})$ for which there exist constants $M_1, M_2, a, k>0$ such that 
  \begin{align}
      p(\m{x}) &< M_1 e^{-a \Norm{\m{x}}}\\
      \Norm{\nabla \log p(\m{x})}^2 &< M_2 \Norm{\m{x}}^k
  \end{align}
  when $\Norm{\m{x}}>L$ for large enough $L$. Then for any $\eps>0$, there exists a recurrent neural network whose firing-rate dynamics are given by \eqref{eq:out_fp}, whose recurrent weights, output weights and the diffusion coefficient are given by $\wrec\in \R^{n\times n}$ of rank $m$, $\wo \in \R^{m\times n}$, and $\sig\in \R^{n\times m}$ respectively, such that, for a large enough $n$, the score of the stationary distribution of the output units $s_\tet(\m{x})$ satisfies $\bbE_{\m{x}\sim p(\m{x})}[\Norm{\nabla \log p(\m{x}) - s_\tet(\m{x})}^2]<\eps$. 
\end{thm*}

\begin{proof}
  $ $\newline
  \begin{description}
    \item[Step 1] 
      First of all, we would like to find a compact set $U_C = \set{\m{x}|\Norm{\m{x}}_2\leq C}$ such that $\int_{\R^m\setminus U_C} \Norm{\nabla \log p(\m{x}) - s_\tet(\m{x})}^2 p(\mx) d\mx<\f{\eps}{2}$. We can upper bound the integral by the following calculation using spherical coordinates when $C>L$. 
      \begin{align}
          \int_{\R^m\setminus U_C} \Norm{\nabla \log p(\m{x}) - s_\tet(\m{x})}^2 p(\mx) d\mx &= \int_{\R^m\setminus U_C} \Norm{\nabla \log p(\m{x})}^2 p(\mx) d\mx \label{eq:compact_support}\\
          &< \int_{\R^m\setminus U_C} M_1 M_2 e^{-a\Norm{\mx}} \Norm{\mx}^k d\mx\\
          &= \f{2 M_1 M_2 \pi^{m/2}}{\Gam(\f{m}{2})}\int_C^\inf r^{m+k-1}e^{-a r} dr\\
          &= \f{2 M_1 M_2 \pi^{m/2}}{a^{m+k}\Gam(\f{m}{2})}\int_{aC}^\inf u^{m+k-1}e^{-u} du\\
          &= \f{2 M_1 M_2 \pi^{m/2}}{a^{m+k}\Gam(\f{m}{2})}\Gam(m+k, aC)
      \end{align}
      
      Note that we assume $s_\tet(\mx) = \m{0}$ outside of $U_C$ \footnote{Equality \eqref{eq:compact_support} is the only place where this assumption is used. The equality also holds if $s_\tet(\mx)$ grows at a similar rate to $\nabla \log p$.}. By lemma \ref{lem:gam_lim}, there exists a large enough constant $\tC$ such that $\Gam(m+k, a\tC)< \f{a^{m+k}\Gam(\f{m}{2})\eps}{4 M_1 M_2 \pi^{m/2}}$. Therefore 
      \begin{equation}
        \int_{\R^m\setminus U_\tC} \Norm{\nabla \log p(\m{x}) - s_\tet(\m{x})}^2 p(\mx) d\mx<\f{2 M_1 M_2 \pi^{m/2}}{a^{m+k}\Gam(\f{m}{2})}\cdot\f{a^{m+k}\Gam(\f{m}{2})\eps}{4 M_1 M_2 \pi^{m/2}} = \f{\eps}{2}
      \end{equation}

    \item[Step 2] Next we write down the specific form of $s_\tet(\mx)$, the score of the stationary distribution of $\mx$ for the reservoir-sampler dynamics \eqref{eq:out_fp}
      \begin{equation}
        s_\tet(\mx) = 2(\wo \sig\sig^T\wo^T)^{-1}(-\alp \mx + \alp \wo \phi(\twr\mx + I)).
      \end{equation}
      Since $\sig$ is learnable, we can define $\sig = \wo^T(\wo \wo^T)^{-1}$ such that $\wo \sig\sig^T\wo^T$ is an identity matrix \footnote{A technical subtlety is that $\wo \wo^T$ is not necessarily invertible. However, since every transformation considered here is continuous, we could add a small perturbation to $\wo$ such that $\wo \wo^T$ is invertible, and our result still holds.}:
      \begin{equation}
        \begin{split}
          \wo \sig\sig^T\wo^T &= \wo \wo^T(\wo \wo^T)^{-1} (\wo \wo^T)^{-1} \wo \wo^T\\
          &= \m{I}
        \end{split}
      \end{equation}
      As $p$ is continuously differentiable, by Theorem \ref{thm:uat}, there exists $\wo, \twr$ such that 
      \begin{equation} \label{eq:approx}
        \sup_{\mx\in U_\tC} \Norm{\wo\phi(\twr\mx + I) - \left(\mx + \f{\nabla \log p(\mx)}{2\alp}\right)}^2<\f{\eps}{8\alp^2}
      \end{equation}
      Therefore with those choices of $\wo, \twr$ and $\sig$, we have
      \begin{equation}
        \begin{split}
          \int_{U_\tC} \Norm{\nabla \log p(\m{x}) - s_\tet(\m{x})}^2 p(\mx) d\mx &< \sup_{\mx\in U_\tC} \Norm{\nabla \log p(\m{x}) - s_\tet(\m{x})}^2\\
          &=\sup_{\mx\in U_\tC} \Norm{\nabla \log p(\m{x}) - 2(-\alp \mx + \alp \wo \phi(\twr\mx + I))}^2\\
          &=4\alp^2 \sup_{\mx\in U_\tC} \Norm{\f{\nabla \log p(\m{x})}{2\alp} + \mx - \wo \phi(\twr\mx + I)}^2\\
          &<4\alp^2 \cdot\f{\eps}{8\alp^2}\\
          &= \f{\eps}{2}
        \end{split}
      \end{equation}
    \item[Step 3] 
    \begin{equation}
      \begin{split}
        \bbE_{\m{x}\sim p(\m{x})} [\Norm{\nabla \log p(\m{x}) - s_\tet(\m{x})}^2] &= \left(\int_{\R^m\setminus U_{\tC}} +\int_{U_{\tC}}\right)p(\m{x}) \Norm{\nabla \log p(\m{x}) - s_\tet(\m{x})}^2 d\m{x} \\
        &<\f{\eps}{2} + \f{\eps}{2} = \eps \qedhere
      \end{split}
      \end{equation}
  \end{description}
\end{proof}

\section{Pseudocode for training the reservoir-sampler network} \label{app:train_code}
\begin{algorithm}
  \caption{Training RSN}
  \KwIn{Training samples $\set{\mx_i}$, pseudo recurrent weights $\twr$, output weights $\wo$, start noise level $\sig_1$, noise decay factor $C$, number of noise level $N$}
  \KwOut{trained recurrent weights $\wrec$, trained output weights $\wo$, trained diffussion coefficient $\sig$}
  \For{$i = 1, \dots, N$}{
    \While{not done}{
      Sample $\mx$ from the training set\\
      Perturb the samples with Gaussian noise $\widetilde{\mx} = \mx_j + \bs{\eps}$, $\bs{\eps} \sim  \cN(0, \sig_i^2\cI)$\\
      Compute the gradient of the score-matching loss $\f{1}{2}\nabla_{\twr, \wo} \Norm{2\alp(\wo \phi(\twr \mx + I) -\mx) - \f{\mx - \widetilde{\mx}}{\sig_i^2}}^2$\\
      Update $\twr, \wo $ with gradient descent
    }
    $\sig_{i+1} = C\sig_i$
  }
  $\wrec \leftarrow \twr \wo$\\
  $\sig \leftarrow \wo^T(\wo\wo^T)^{-1}$\\
  \Return $\twr, \wo$
\end{algorithm}

\section{Accelerated sampling with reservoir-sampler network} \label{app:accel}
To construct an irreversible dynamics that accelerates sampling, we let the divergence-free (DF) field G be of the form $J\nabla p$, where $J = -J^T$ can be any skew-symmetric matrix that is either learned through training or prescribed. Then from equation \eqref{eq:fp_stat} we see that the drift term $F_\tet$ satisfies
\begin{equation}
F_\tet(\m{v}) = \Sig \nabla \log p(\m{v}) + J\nabla \log p(\m{v}) = (\Sig + J) \nabla \log p(\m{v}).
\end{equation}
Therefore the corresponding score matching problem is $\alpha (\Sig + J)^{-1}(W_{out} \phi(\tilde{W}_{rec} x + I) - x) \approx \nabla\log p(x)$. With a slight modification, the pseudo code for an accelerated RSN training algorithm is given in algorithm \ref{algo:arsn}.
\begin{algorithm}\label{algo:arsn}
  \caption{Training accelerated RSN}
  \KwIn{Training samples $\set{\mx_i}$, pseudo recurrent weights $\twr$, output weights $\wo$, pseudo diffusion coefficient $\tilde{\sig}$, skew-symmetric matrix $J$, start noise level $\sig_1$, noise decay factor $C$, number of noise level $N$}
  \KwOut{trained recurrent weights $\wrec$, trained output weights $\wo$, trained diffussion coefficient $\sig$}
  \For{$i = 1, \dots, N$}{
    \While{not done}{
      Sample $\mx$ from the training set\\
      Perturb the samples with Gaussian noise $\widetilde{\mx} = \mx_j + \bs{\eps}$, $\bs{\eps} \sim  \cN(0, \sig_i^2\cI)$\\
      Compute the gradient of the score-matching loss $\f{1}{2}\nabla_{\twr, \wo, \tilde{\sig}, J} \Norm{\alp (\f{1}{2}\tilde{\sig}\tilde{\sig}^T + J)^{-1}(\wo \phi(\twr \mx + I) -\mx) - \f{\mx - \widetilde{\mx}}{\sig_i^2}}^2$\\
      Update $\twr, \wo, \tilde{\sig}$ with gradient descent\\
      Update $J$ with gradient descent (Optional)
    }
    $\sig_{i+1} = C\sig_i$
  }
  $\wrec \leftarrow \twr \wo$\\
  $\sig \leftarrow \wo^T(\wo\wo^T)^{-1} \tilde{\sig}$\\
  \Return $\twr, \wo$
\end{algorithm}
We repeat the experiment in Section \ref{subsec:bimodal} with the accelerated RSN. The results are shown in Figure \ref{fig:accel}. We see that the accelerated RSN is able to sample from the distribution faster than the RSN. 

\begin{figure}
  \centering
  \includegraphics[width = \textwidth]{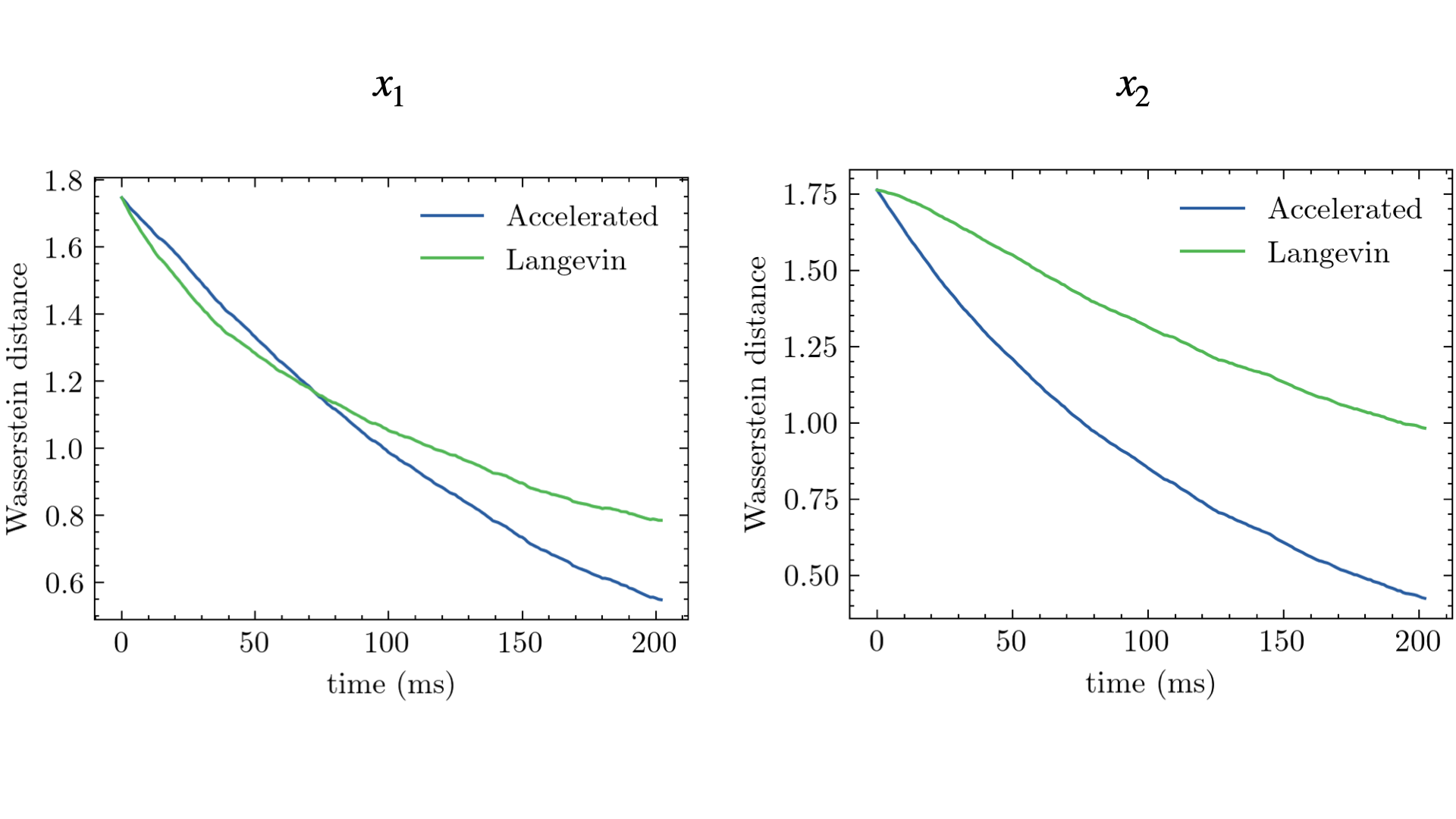}
  \caption{\textbf{Sampling speed of accelerated RSN} We repeat the experiment described in Section \ref{subsec:bimodal} where we try to sample from a 2-D Laplacian mixture distribution. The figure shows the Wasserstein distance between the marginalized distributions of generated samples and true samples for both dimensions. We see that compared to vanilla RSN, accelerated RSN is able to sample from the stationary distribution with faster relaxation time. }
  \label{fig:accel}
\end{figure}
\section{Additional experiments} \label{app:exp}
\subsection{Failure cases for sampler-only networks with hyperbolic tangent nonlinearity} \label{app:tanh_failure}
Figure \ref{fig:tanh_succeeds} shows that with tanh nonlinearity, sampler-only network is able to sample from the distribution we consider in Section \ref{subsec:bimodal}. However, this is because the score function of the distribution considered can be spanned by just two basis functions equipped by the drift term of the sampler-only network. Therefore, our theory predicts that the sampler-only network will not be able to approximate score functions that are not exactly spanned by these two basis functions. Indeed, if we consider a slightly different Gaussian mixture distribution, $p_{\mrm{data}}(x) = \f{1}{2} (\cN(-1, 0.0625) + \cN(1, 0.25))$, then the sampler-only network will not be able to recover the score function as shown in Figure \ref{fig:tanh_fail}, while RSN still succeeds in recovering the score function.
\begin{figure}[htbp]
  \centering
  \includegraphics[width = \textwidth]{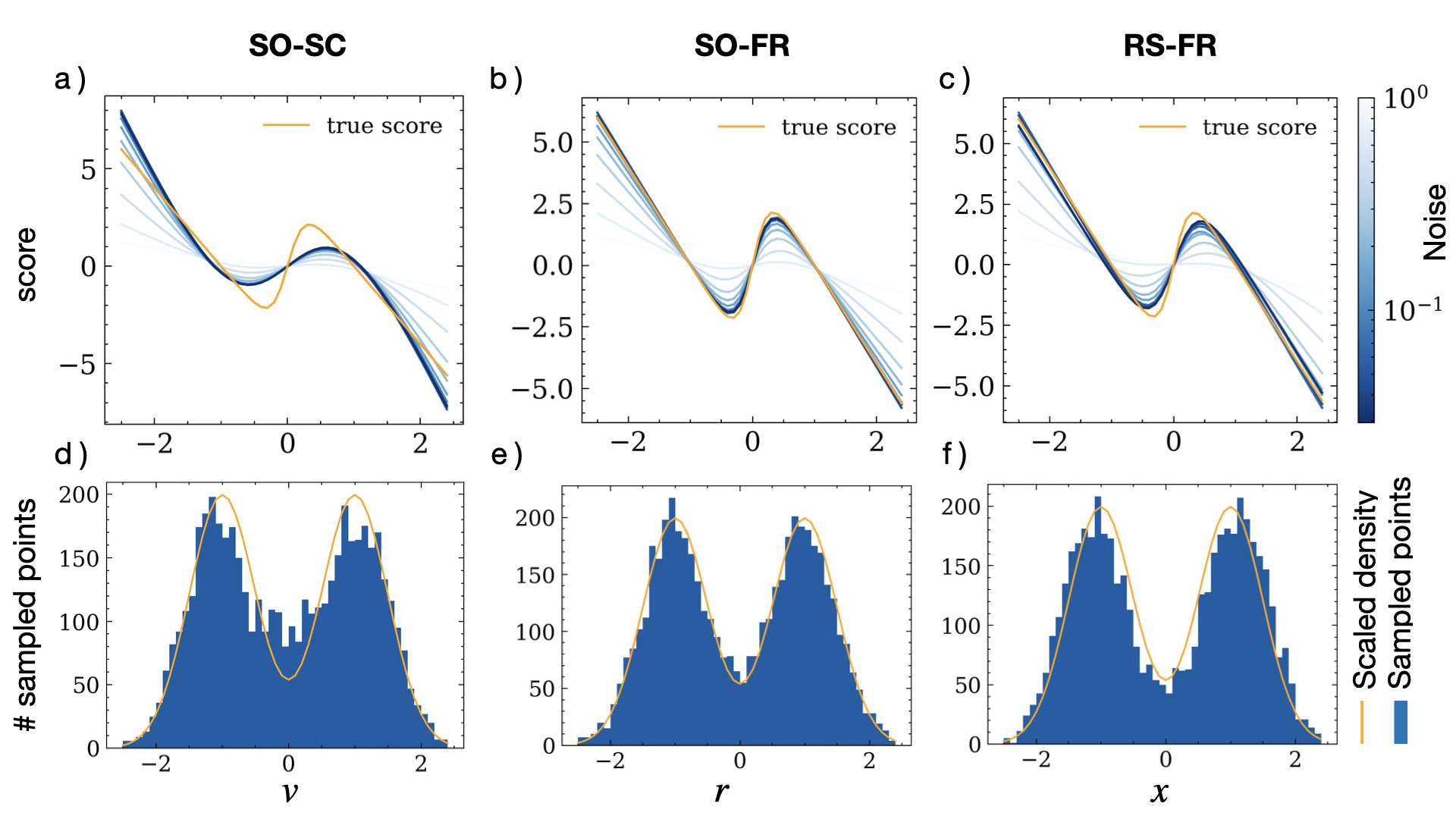}
  \caption{\textbf{Sampler-only networks with hyperbolic tangent nonlinearity are able to sample from a 1-D Gaussian mixture distribution with a symmetric score function.}. See Figure \ref{fig:DP} for an explanation of the abbreviations. Here we use tanh nonlinearity for sampler-only networks (SO-SC and SO-FR). All 3 types of networks are able to match the score function and sample from the distribution.}
  \label{fig:tanh_succeeds}
\end{figure}
\begin{figure}[htbp]
  \centering
  \includegraphics[width = \textwidth]{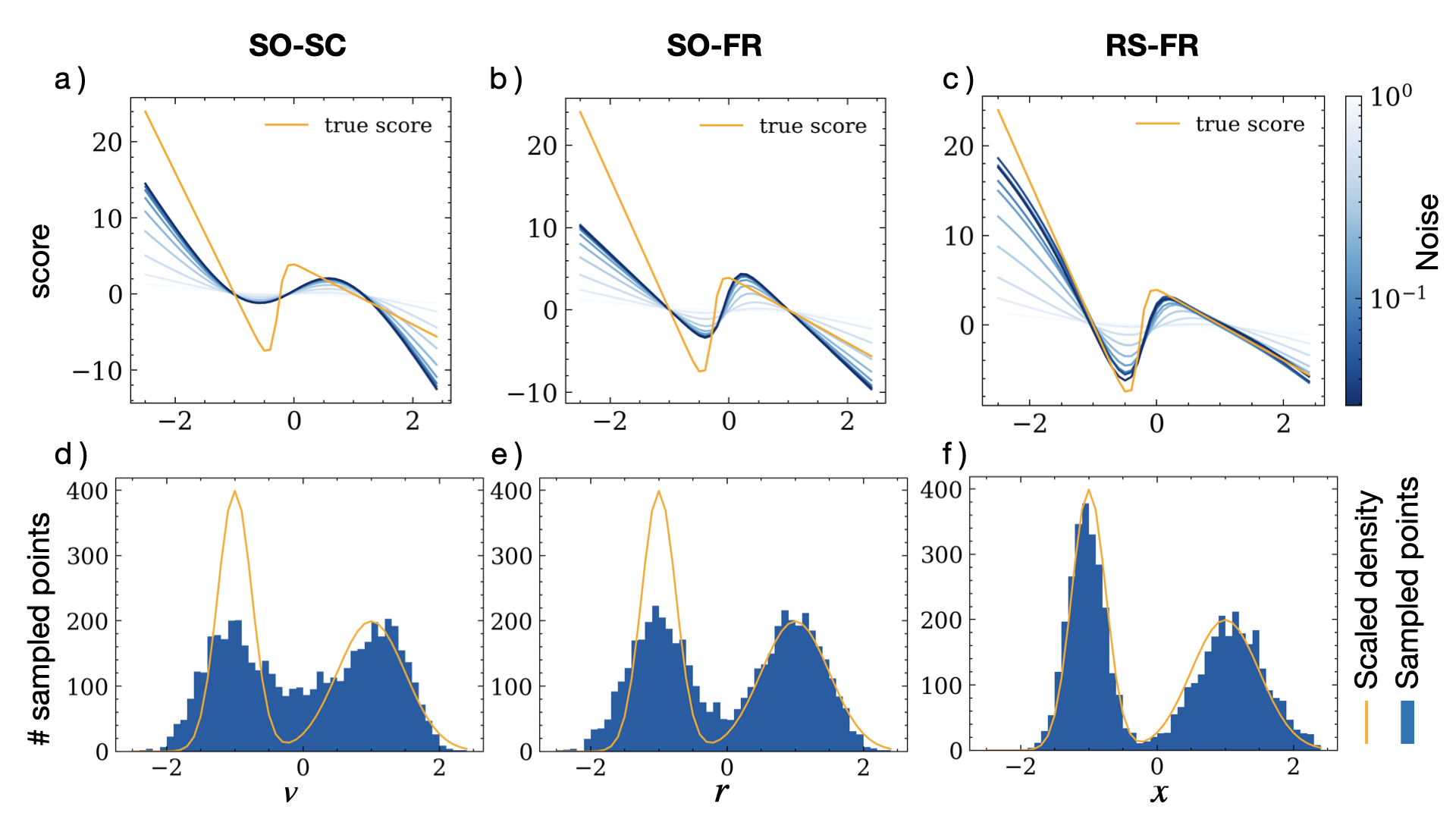}
  \caption{\textbf{Sampler-only network with hyperbolic tangent nonlinearity fails to sample from a 1-D Gaussian mixture distribution with a non-symmetric score function.} See Figure \ref{fig:DP} for an explanation of the abbreviations. Here we use tanh nonlinearity for sampler-only networks (SO-SC and SO-FR). Only RSN is able to match the score function and sample from the distribution.}
  \label{fig:tanh_fail}
\end{figure}

\subsection{Sampling using autoencoder}
We show that we can replace the fixed linear PCA projection used in Section \ref{subsec:MNIST} with pretrained nonlinear autoencoder. We test the algorithm on both the MNIST dataset (Figure \ref{fig:ae_mnist}) and the CIFAR-10 dataset (Figure \ref{fig:ae_cifar}).
\begin{figure}[htbp]
  \centering
  \includegraphics[width = \textwidth]{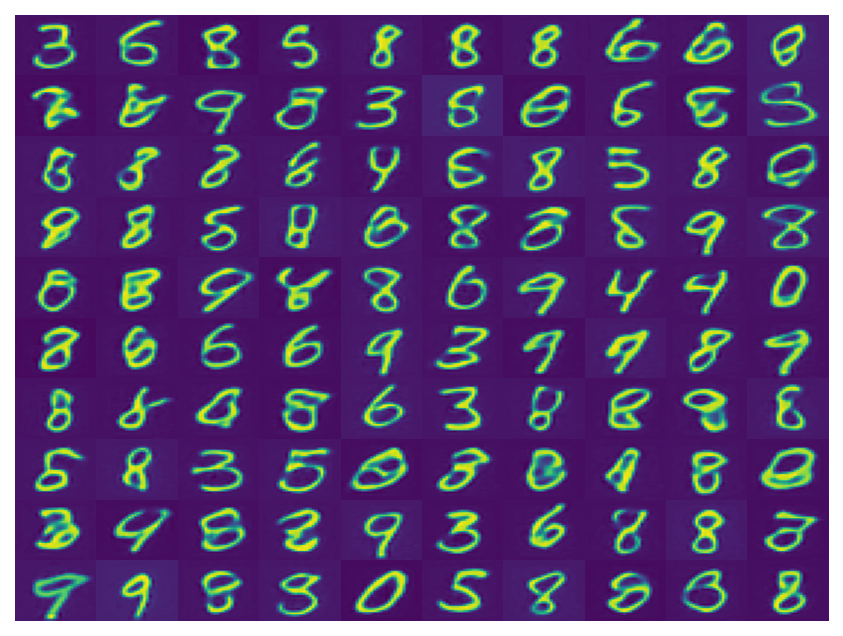}
  \caption{Samples produced by our reservoir sampler network when trained to fit the latent distribution produced by a non-linear convolution autoencoder from MNIST. The distribution was 32-dimensional. The reservoir sampler network had 500 reservoir neurons.}
  \label{fig:ae_mnist}
\end{figure}
\begin{figure}[htbp]
  \centering
  \includegraphics[width = \textwidth]{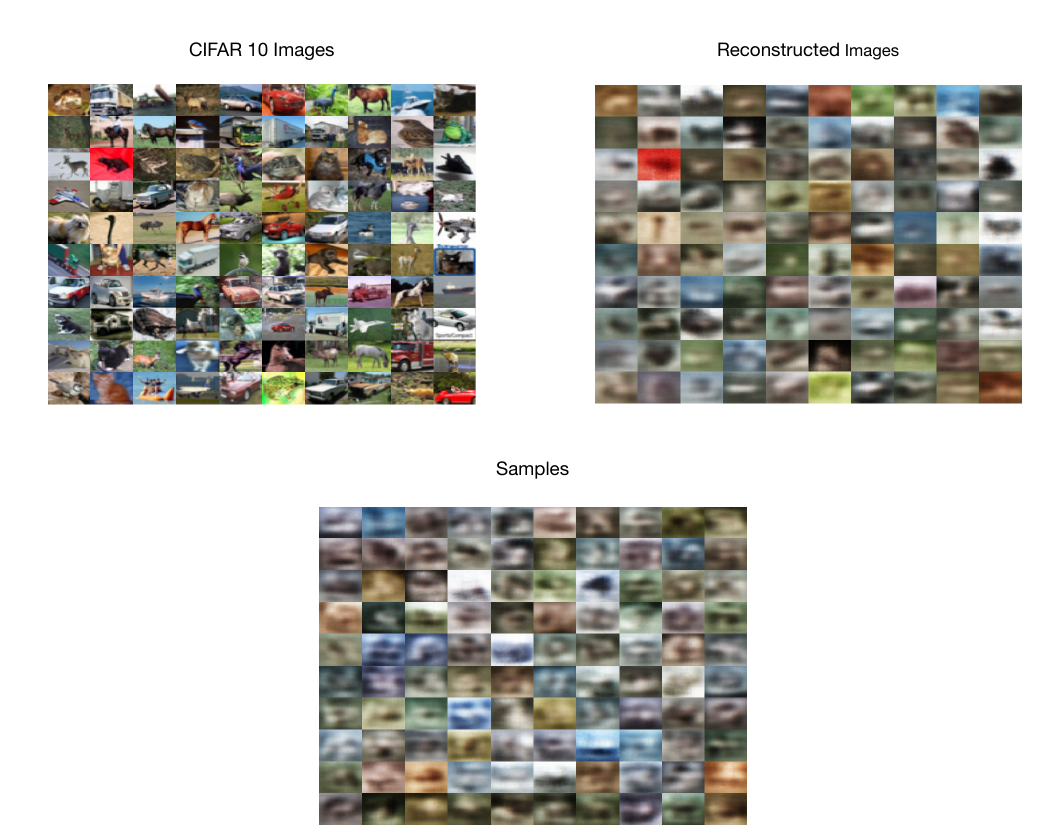}
  \caption{Samples produced by our reservoir sampler network when trained to fit the latent distribution produced by a non-linear convolution autoencoder from CIFAR-10. We show both the original images and the reconstructed images to evaluate the sampling quality. The distribution was 128-dimensional. The reservoir sampler network had 20000 reservoir neurons.}
  \label{fig:ae_cifar}
\end{figure}
\section{Related works} \label{app:equivalence}
In this section, we examine the neural dynamics considered in \citep{dong_adaptation_2022, aitchison_hamiltonian_2016, masset_natural_2022, echeveste_cortical-like_2020} and discuss how our work is related. 

\citet{dong_adaptation_2022} uses Hamiltonian Langevin dynamics to accelerate sampling in continuous attractor networks. The Hamiltonian dynamics below are considered in the paper (cf. eq (19) and (20)):
\begin{equation}\label{eq:HLD}
  \begin{cases}
    \tau_s \f{d\m{s}}{dt} &= \bs{\alp} \m{y}\\
    \tau_z \f{d\m{y}}{dt} &= -\bs{\beta \alp^{-1}} \m{y} + \Lambda (\m{s^o - s}) + \sqrt{\tau_z}\bs{\sig_y \xi}_t.
  \end{cases}
\end{equation}
The overdamped Langevin dynamics that samples from the same marginal stationary distribution over $\m{s}$ as \eqref{eq:HLD}:
\begin{equation}
\f{d\m{s}}{dt} =  \Lambda (\m{s^o - s})  + \sqrt{2}\bs{\xi}_t,
\end{equation}
where $\m{s^o}$ is an observation generated by a latent feature $\m{s}$ according to a Gaussian distribution, and the linear term $\Lam(\m{s^o - s})$ comes from the external input. So the network dynamics in \citet{dong_adaptation_2022} themselves do not generate samples without the external input, and the study does not clarify how to supply the network with nonlinear score functions. 


\citet{aitchison_hamiltonian_2016} have shown that Hamiltonian Monte Carlo (HMC) naturally maps onto the dynamics of excitatory-inhibitory neural networks. However, the neural dynamics considered are entirely linear in terms of the network variables (cf. eq (23) and (24)). Therefore, the dynamics are again equivalent to the synaptic current dynamics with an identity transfer function in Section \ref{subsec:synap}. The same applies to the analysis in \citet{masset_natural_2022}, which also only considers linear neural dynamics. By Proposition \ref{prop:lim_exp} and the decomposition discussed in \citet{kwon_structure_2005}, these dynamics can only sample from probability distributions whose score functions are linear, i.e. Gaussians. 

The work by \citet{echeveste_cortical-like_2020} provides a unifying model for several dynamical phenomena in sensory cortices by training a nonlinear recurrent excitatory-inhibitory neural circuit model to perform sampling-based probabilistic inference. The synaptic neural dynamics below (cf. equation (8)) is used, where the nonlinear transfer function is supralinear, i.e. $\phi(u) = k \ceil{u}^n$:
\begin{equation} \label{eq:echeveste}
  \tau_i \f{du_i}{dt} = -u_i + \sum_{j=1}^N w_{ij} \phi(u_j) + h_i(t) + \eta_i(t).
\end{equation} 
 As discussed in Section \ref{subsec:synap}, while they are richer than linear dynamics, dynamics of this form still have limited ability to sample from complex probability distributions. Moreover, we note that the training criterion considered in \citet{echeveste_cortical-like_2020} only matches the mean and the variance of the stationary distribution. These points are in contrast to our work here, where we show that the reservoir-sampler network, which extends the dynamics \eqref{eq:echeveste}, can be trained via alternate methods to sample from complex probability distributions by matching their score functions.

\section{Details of numerical experiments} \label{app:train_detail}
All experiments were run on one NVIDIA Quadro RTX 6000 GPU. We always use the model saved at the last training iteration. 
\subsection{Training details for Section \ref{subsec:bimodal}}
For the 1-D bimodal Gaussian distribution experiment (Figure \ref{fig:DP}), we used 1000 reservoir neurons (only in the reservoir sampler network) and 1 sampler neuron. We trained all 3 networks (SO-FR, SO-SC, RS-FR) for 79000 iterations with the Adam optimizer. The learning rate was 0.0001, and the batch size was 128. We started by adding Gaussian noise with standard deviation (std) $\sig_1 = 1$ to the training samples. We decreased the std every 7900 iterations by a constant factor, i.e. $\sig_{i+1}/\sig_{i} = C$, such that $\sig_{10} = .01$. When simulating the stochastic neural dynamics, we used a step size of $10^{-4}$, and ran 10000 steps. 

For the mixture of 2-D Laplace distributions experiment (Figure \ref{fig:lap}), we used 1000 reservoir neurons (only in the reservoir sampler network) and 2 sampler neurons. We trained the RS-FR network for 62,500 iterations with the Adam optimizer. The learning rate was 0.0001, and the batch size was 128. We started by adding Gaussian noise with std $\sig_1 = 1$ to the training samples. We decreased the std every 6250 iterations by a constant factor, i.e. $\sig_{i+1}/\sig_{i} = C$, such that $\sig_{10} = .01$. When simulating the stochastic neural dynamics, we used a step size of $10^{-4}$, and ran 20000 steps.

\subsection{Training details for Section \ref{subsec:MNIST}}
For the MNIST experiment (Figure \ref{fig:DP}), we used 20000 reservoir neurons (only in the reservoir sampler network) and 300 sampler neurons. We trained all 3 networks (SO-FR, SO-SC, RS-FR) for 1000 epochs with the Adam optimizer. The learning rate was 0.001, and the batch size was 64. We started by adding Gaussian noise with std $\sig_1 = 1$ to the training samples. We decreased the std every 100 epochs by a constant factor, i.e. $\sig_{i+1}/\sig_{i} = C$, such that $\sig_{10} = .01$. When simulating the stochastic neural dynamics, we used a step size of $10^{-6}$, and ran 30000 steps. 


        
\section{Statement of code availability}
All code is available on \href{https://github.com/chinsengi/score_RNN}{Github} 

\end{document}